\documentclass[11pt]{article}

\usepackage{geometry}
\usepackage[ruled]{algorithm2e} 

\SetAlFnt{\small}
\SetAlCapFnt{\small}
	\SetAlCapNameFnt{\small}
\SetAlCapHSkip{0pt}
\IncMargin{-1.6\parindent}

\usepackage{amsmath,amssymb,amsthm}
\usepackage{enumerate}
\usepackage[colorlinks=true,citecolor=blue,linkcolor=blue,urlcolor=blue]{hyperref}
\usepackage{mathtools}
\usepackage{tikz}
\usepackage{url}
\usepackage{pifont}
\usetikzlibrary{positioning,chains,fit,shapes,calc}

\newcommand{\R}{\mathbb{R}}
\newcommand{\Q}{\mathbb{Q}}
\newcommand{\Z}{\mathbb{Z}}

\newcommand{\Si}{\mathcal{S}_i}

\usepackage{multirow}
\usepackage{thmtools,thm-restate}
\usepackage{cleveref}

\usepackage{ifthen}
\newenvironment{rtheorem}[3][]{%
\noindent\ifthenelse{\equal{#1}{}}{\bf #2 #3.}{\bf #2 #3 (#1)}%
\begin{it}}{\end{it}}

\newtheorem{theorem}{Theorem}
\newtheorem{lemma}[theorem]{Lemma}
\newtheorem{corollary}[theorem]{Corollary}
\newtheorem{proposition}[theorem]{Proposition}

\newtheorem{definition}[theorem]{Definition}
\newtheorem{example}[theorem]{Example}

\newtheorem{remark}[theorem]{Remark}

\numberwithin{theorem}{section}
\numberwithin{lemma}{section}
\numberwithin{proposition}{section}
\numberwithin{remark}{section}
\numberwithin{example}{section}
\numberwithin{corollary}{section}
\numberwithin{figure}{section}
\numberwithin{definition}{section}

\begin{document}
\title{Sampling from the Gibbs Distribution in Congestion Games\footnote{Most of this work has been carried out while the author was a postdoctoral fellow at the Max Planck Institute for Informatics in Saarbr\"ucken, Germany.}}

\author{
Pieter Kleer\\ 
Tilburg University\\
Tilburg, The Netherlands\\ 
\texttt{p.s.kleer@tilburguniversity.edu}
}
\maketitle              
\begin{abstract} 
Logit dynamics is a form of randomized game dynamics where players have a bias towards strategic deviations that give a higher improvement in cost. It is used extensively in practice. In congestion (or potential) games, the dynamics converges to the so-called Gibbs distribution over the set of all strategy profiles, when interpreted as a Markov chain. In general, logit dynamics might converge slowly to the Gibbs distribution, but beyond that, not much is known about their algorithmic aspects, nor that of the Gibbs distribution. 
In this work, we are interested in the following two questions for congestion games: i) Is there an efficient algorithm for \emph{sampling} from the Gibbs distribution? ii) If yes, do there also exist natural randomized dynamics that converges quickly to the Gibbs distribution?

We first study these questions in  extension parallel congestion games, a well-studied special case of symmetric network congestion games. As our main result, we show that there is a simple variation on the logit dynamics (in which we in addition are also allowed to randomly interchange the strategies of two players) that converges quickly to the Gibbs distribution in such games. This answers both questions above affirmatively. We also address the first question for the class of so-called capacitated $k$-uniform congestion games.

To prove our results, we rely on the recent breakthrough work of Anari, Liu, Oveis-Gharan and Vinzant (2019) concerning the approximate sampling of the base of a matroid according to strongly log-concave probability distribution.
\end{abstract}

\section{Introduction} 
Congestion games constitute a rich class of games that have been studied extensively since their introduction by Rosenthal \cite{Rosenthal1973}. An (unweighted) \emph{congestion game} $\Gamma = (N,E,(\Si)_{i \in N},(c_e)_{e \in E})$ consists of a set of players $N = \{1,\dots,n\}$ and a set of resources $E = \{1,\dots,m\}$. Every player $i$ has a strategy set $S_i \subseteq 2^E$, where each strategy is a subset of resources. Furthermore, every resource $e \in E$ is equipped with a cost function $c_e : \R_{\geq 0} \rightarrow \R$ that we assume to be non-negative and non-decreasing. 
The goal of a player is to choose a strategy that minimizes her total cost $C_i(s) = \sum_{e \in s_i} c_e(\ell_e(s))$, where $\ell_e(s)$ is the number of players using resource $e$ in profile $s \in \times_i \Si = \mathcal{S}$.  
A well-known example is the class of symmetric network congestion games, in which we are given a directed graph $G = (V,E)$ with origin $o \in V$ and destination $d \in V$. The common strategy set of all players is given by the set of all $o,d$-paths in $G$.

Rosenthal \cite{Rosenthal1973} proved that congestion games are \emph{(exact) potential games}. He showed that the function $\Phi : \times_i \Si \rightarrow \R$ given by \vspace*{-0.15cm}
$$
\Phi(s) = \sum_{e \in E} \sum_{k=1}^{\ell_e(s)} c_e(k),
$$
satisfies 
\begin{equation}
\label{eq:exact}
C_i(s) - C_i(s_i',s_{-i}) = \Phi(s) - \Phi(s_i',s_{-i})
\end{equation}
for every $s \in \times_i\Si$ and $s_i' \in \mathcal{S}_i$. Here $C_i(s_i',s_{-i})$ is used to denote the cost of player $i$ in the strategy profile where $i$ chooses $s_i'$, and all other players choose their strategy in $s$. The function $\Phi$ is often referred to as \emph{Rosenthal's potential}.
The main implication of \eqref{eq:exact} is the existence of a \emph{pure Nash equilibrium (PNE)}: A strategy profile in which no player can deviate to another strategy and obtain an improved cost \cite{Rosenthal1973}. This follows directly from the observation that any sequence of \emph{better (or best) response dynamics} converges to a PNE in a finite number of steps. Better response dynamics is defined as the procedure where in every step precisely one player deviates to another strategy that yields an improved cost (until a pure Nash equilibrium is reached). For best response dynamics, the deviating player always deviates to a strategy that yields the greatest possible improvement in cost.  

 In the last two decades, the algorithmic aspects of (pure) Nash equilibria have been studied extensively, in both general and special classes of congestion games; see, e.g., \cite{Fabrikant2004,Ieong2005,Ackermann2008,DelPia2015,KleerSchaferEC}. Two of the most prominent questions concerning pure Nash equilibria are the following.
\begin{enumerate}
\item Do (natural) player dynamics, such as better or best response dynamics, converge to a PNE in polynomial time?
\item If not, can one compute a PNE in polynomial time by other means?
\end{enumerate}

\noindent Player dynamics, roughly speaking, come in two flavors: Either one deviates to another strategy profile according to a deterministic rule, or a probabilistic one. A well-known example of the latter case is \emph{noisy (randomized) best response dynamics} that has received a lot of attention in practice, but seems hard to analyze from a theoretical perspective. Here, instead of making a deviation  to another strategy according to a deterministic rule, a player chooses a strategy from her set according to a probability distribution that usually puts relatively more weight on strategies that will result in a lower cost.\footnote{Such methods are special cases of what is more broadly known as \emph{simulated annealing}.}  Randomized dynamics can be studied from multiple perspectives: Either as a randomized alternative for deterministic dynamics converging to a pure Nash equilibrium or as a dynamical system on its one. 

One well-known example of player dynamics that can be studied as a dynamical system, and which is the topic of this paper, is the \emph{logit dynamics}. It has  received a lot of attention in various  communities, such as evolutionary game theory (see, e.g., \cite{Sandholm2010}) and experimental economics (see, e.g., \cite{Camerer2010}). The procedure was introduced by Blume \cite{Blume1993} as a form of randomized game dynamics in which players update their strategy according to a \emph{logit update rule} \cite{McFadden1973}. The logit dynamics for congestion games can be formulated as follows. For a given strategy profile $s$ and fixed \emph{rationality level} (or \emph{temperature}\footnote{The notion of temperature comes from the physics literature.}) parameter $T \geq 0$, first choose a player $i \in N$ uniformly at random, and then have player $i$ choose a strategy $s_i' \in \mathcal{S}_i$ with probability 
\begin{equation}
\label{eq:logit}
\frac{e^{-T \Phi(s_i',s_{-i})}}{\sum_{t \in \mathcal{S}_i}  e^{-T \Phi(t,s_{-i})}}.\footnote{Equivalently, and in fact the usual definition of logit dynamics for general (not necessarily potential) games, one can replace the $\Phi(s_i',s_{-i})$ by $C_i(s_i',s_{-i})$, and, similarly, in the normalizing constant. We choose to use $\Phi$ as this will be more convenient for our purposes. The equivalence follows from \eqref{eq:exact}.}
\end{equation}
 Note that the denominator in \eqref{eq:logit} is a normalizing constant. The rationality level $T \geq 0$ is used to model the amount of noise players believe there to be in the system. When $T \rightarrow \infty$, players effectively only assign positive probability to best responses, whereas when $T \rightarrow 0$, the distribution in \eqref{eq:logit} approaches the uniform distribution over $\Si$. Also note that the dynamics indeed puts relatively more probability mass on strategies that give a greater improvement in cost.

The logit dynamics gives rise to an \emph{irreducible, time-reversible Markov chain} on the set $\mathcal{S}$ of all strategy profiles that has the \emph{Gibbs distribution} $\pi$, given by
$$
\pi(s) = \frac{e^{-T \Phi(s)}}{\sum_{t \in \times_i \mathcal{S}_i} e^{-T \Phi(t)}}
$$
for $s \in \mathcal{S}$, as its unique stationary distribution. This simply  means that if one runs the logit dynamics for a sufficiently long time, the distribution over $\mathcal{S}$ converges to the Gibbs distribution.

In fact, Auletta et al. \cite{Auletta2016} interpret the Gibbs distribution as a dynamic equilibrium concept, which they dubbed the \emph{logit equilibrium} (this concept is well-defined for general games in \cite{Auletta2016}). 
The goal of this work is to study algorithmic aspects of the logit equilibrium/Gibbs distribution in congestion games. Unfortunately, in general the logit dynamics converge slowly to the Gibbs distribution  \cite{Auletta2016}, in particular, the number steps needed might be $\Omega(e^{T\Phi_{\max}})$, where $\Phi_{\max}$ is the maximum value attained by Rosenthal's potential. We next give a simple example illustrating this fact.

\begin{example}
Consider the congestion game with players $N = \{1,2\}$ and two resources $E = \{a,b\}$. Both resources $e \in \{a,b\}$ can be used by both players, i.e., $\mathcal{S}_1 = \mathcal{S}_2 = \{\{a\},\{b\}\}$, and have a cost function satisfying $c_e(1) = 0$ and $c_e(2) = \phi$ for some $\phi \geq 0$. 

If $\phi$ is large, the Gibbs distribution assigns weight close to $1/2$ to the strategy profiles $(s_1,s_2) \in \{(a,b),(b,a)\}$, and weight close to $\frac{1}{2}e^{-T\phi} \approx 0$ to $(s_1,s_2) \in \{(a,a),(b,b)\}$. Now, informally speaking, if we consider the logit dynamics with starting profile $(a,b)$, then the number of steps we need to run the dynamics before we see the profile $(b,a)$ is $\Omega(e^{T\phi})$. As both profiles appear with probability close to $1/2$ in the Gibbs distribution, this also means that we need to run the dynamics for at least $\Omega(e^{T\phi})$ steps before the resulting distribution over all strategy profiles is close to the Gibbs distribution.\label{exmp:slow}
\end{example}

What is causing the slow convergence in Example \ref{exmp:slow}? The problem is that we need to use either the profile $(a,a)$ or $(b,b)$, both having very small probability in the Gibbs distribution, to move from $(a,b)$ to $(b,a)$. However, as it turns out, if one in addition with some probability is allowed to interchange the strategies of players $1$ and $2$, then the resulting dynamics converges  quickly to the Gibbs distribution.\footnote{This is a special case of Theorem \ref{thm:ep_informal}.} Note that this enables the possibility to directly transition between $(a,b)$ and $(b,a)$. This motivates the following question:
\begin{center}
\textit{Is there a (simple) Markov chain on $\mathcal{S}$ that converges rapidly to the Gibbs distribution over $\mathcal{S}$ at any temperature $T$?}
\end{center}
This question may be interpreted as the  natural analogue of looking for other local search procedures converging quickly to a PNE, when best/better response dynamics does not have this property. 

When the answer to the above question is not directly obvious, one can take another step back and first ask whether it is at all possible to efficiently \emph{sample} from the Gibbs distribution. Informally speaking, can we take `snapshots' (according to the Gibbs distribution) from the system in equilibrium in polynomial time? More formally speaking: 
\begin{center}
\textit{Does there exist an efficient algorithm to sample (approximately) a strategy profile $s \in \mathcal{S}$   according to the Gibbs distribution over $\mathcal{S}$ at any temperature $T$?}
\end{center}
This question can be interpreted as a dynamic analogue of the second question posed earlier for the computation of pure Nash equilibria. That is, although (deterministic) better/best response dynamics might take a long time to converge to a pure Nash equilibrium, one still wants to know whether a PNE can be computed efficiently by other means. Similarly, although logit (or other natural) dynamics might take a long time to converge to the Gibbs distribution, we may still ask whether, by means of sampling, we can get an impression of what the Gibbs distribution over $\mathcal{S}$ looks like.

The questions above will  be made precise in Section \ref{sec:pre}. We remark that one of the aspects that makes them non-trivial is the fact that we require the questions  to hold for \emph{any} temperature $T$, i.e., $T$ is considered part of the input. (For example, in Example \ref{exmp:slow} we could set $T = \Theta(1/\phi)$ to circumvent the problem arising there.)

\subsection{Our contributions and techniques}\label{sec:contributions}
We first address the questions from the introduction for \emph{extension parallel (EP)} congestion games, a well-studied special case of symmetric network congestion games, see, e.g., \cite{Holzman1997,Fotakis2010,Fujishige2015}. Here, the common strategy set of all players is given by the set of $o,d$-paths $\mathcal{P}$ of an \emph{extension parallel graph} (see Section \ref{sec:ep} for a definition and example). Our main result is that there is a simple Markov chain converging quickly to the Gibbs distribution over $\mathcal{S}$ (also implying that we can sample approximately from the Gibbs distribution). 

We show that if one, in addition to the logit dynamics transitions, is allowed to randomly interchange the strategies of two players (akin to the explanation given after Example \ref{exmp:slow}), the resulting Markov chain converges quickly to the Gibbs distribution. We call this dynamics the \emph{relaxed logit dynamics} (see Section \ref{sec:relaxed} for a formal definition). Note that Theorem \ref{thm:ep_informal} gives a doubly exponential improvement w.r.t. the dependence on $T\Phi_{\max}$ compared to the lower bound as given in Example \ref{exmp:slow}.

\begin{theorem}[Informal]
The \emph{relaxed logit dynamics} for extension parallel (EP) congestion games, at temperature $T$, converges to a distribution ``$\epsilon$-close'' to the Gibbs distribution in at most
$$
n^3 \left( \log n + \log \log |\mathcal{P}| + \log \left(\frac{2T\Phi_{\max}}{\epsilon^2}\right) \right)
$$
steps, where $n$ is the number of players, $\mathcal{P}$ the number of paths in the EP graph and $\Phi_{\max}$ the maximum value attained by Rosenthal's potential.\footnote{If one would drop the assumption that the cost functions are non-negative (see Section \ref{sec:pre}), the parameter $\Phi_{\max}$ can be replaced by $\Delta \Phi := \Phi_{\max} - \Phi_{\min}$ where $\Phi_{\min}$ is the minimum value attained by Rosenthal's potential. This also holds for all subsequent results.} \label{thm:ep_informal}
\end{theorem}

The notion of ``$\epsilon$-close'' refers to the fact that the distribution seen after the indicated number of steps differs from the Gibbs distribution at most $\epsilon$ in \emph{total variation distance} (see Section \ref{sec:markov}), a well-known distance measure for comparing probability distributions in Markov chain theory.

In a nutshell, Theorem \ref{thm:ep_informal} follows from the fact that in EP congestion games Rosenthal's potential is  \emph{M-convex}, as was shown by Fujishige et al. \cite{Fujishige2015}. M-convexity is a property defined in the area of discrete convex analysis \cite{Murota1998} (see Section \ref{sec:matroid}).\footnote{In fact, the result above generalizes directly to ``symmetric congestion games for which Rosenthal's potential is M-convex'', but we are not aware of any other interesting class of congestion games for which this is true (and therefore choose to formulate our results in terms of EP congestion games). M-convexity of Rosenthal's potential already fails to hold for the smallest non-extension parallel network congestion game (the "Figure 8" graph that has two graphs, both consisting of two parallel edges, in series).} The link between M-convexity and sampling has, roughly speaking, been established in a series of papers by Anari et al.  \cite{SLC1,SLC2,SLC4} through the theory of strongly log-concave polynomials, which in turn is also largely developed by Br\"and\'en and Huh \cite{BH2019} (under the name Lorentzian polynomials). In particular, Anari et al. \cite{SLC2} give the first polynomial time algorithm for approximately sampling and counting the number of bases of a given matroid, resolving also an old conjecture by Mihail and Vazirani \cite{Mihail1989}. In this work, we rely on the sampling result from \cite{SLC2}, albeit for relatively simple matroid structures. 

Before proving the result above, we also give another way of sampling from the Gibbs distribution in Section \ref{sec:ep_gibbs}, that essentially is a more direct approach than the sampler induced by the  Markov chain result given above. The high-level approach used for this more direct sampler (and for the additional application given in the next paragraph) is given in Section \ref{sec:overview}. Finally, we give an application of our results to the problem of (almost) uniformly sampling pure Nash equilibria in EP congestion games in Section \ref{sec:nash}, which is, to the best of our knowledge, the first of its kind.\\

\noindent Furthermore, we also study the class of so-called \emph{$u$-capacitated $k$-uniform congestion games} for given $k = (k_1,\dots,k_n)$ and $u = (u_1,\dots,u_m)$. In such a game the strategy set of player $i \in N$ is given by all subsets of $E$ of size $k_i$. Furthermore, for every $e \in E$ we are given a \emph{capacity} $u_e$ so that $c_e(x) = \infty$ whenever $x > u_e$. 

The motivation for studying these games comes from the class of \emph{base-matroid} congestion games, where the strategy set of every player is the set of bases $\mathcal{B}_i$ of a given matroid $\mathcal{M}_i$ over the ground set of resources $E$. It is well-known that best response dynamics converges to a PNE in a polynomial number of steps in this class of games  \cite{Ackermann2008}, and so, in particular, a PNE can be computed in polynomial time. Given the base matroid sampling result of Anari et al. \cite{SLC2}, a natural question that comes to mind is  if a similar result exists for sampling from the Gibbs distribution in base matroid congestion games. Here, we give a first result addressing this question. (Note that in our setting the strategies of player $i \in N$ are the bases of the $k_i$-uniform matroid.)

We next explain why there is a need for capacity constraints in our results. A strategy profile $s$ can be seen as a bipartite graph where the nodes on one side correspond to the players, having degrees $k_i$, and the nodes on the other side  correspond to the resources, having degrees $\ell_e(s)$ (the resource load on $e$ in profile $s$). For a given profile $s$ with resource load profile $\ell(s)$, the bipartite graph is obtained in the natural way: There is an edge between player $i$ and resource $e$ if player $i$ uses resource $e$ in profile $s$. Very roughly speaking, in order to apply the sampling result in \cite{SLC2}, we establish strong log-concavity of a certain polynomial associated to the vectors $k$ and $u$. For this we rely on an asymptotic enumeration formula\footnote{Asymptotic enumeration of graphs with given degrees has been studied extensively in the area of combinatorics.} for the number of bipartite graphs with a given degree sequence, which, in terms of congestion games, gives the number of strategy profiles that have a given resource load profile $\ell$. The formula that we use is only valid for the range of $k = (k_1,\dots,k_n)$ and resource load profiles $\ell(s) = (\ell_e(s))$ satisfying the imposed capacity constraints as given in Theorem \ref{thm:cap_informal}. Our main result is as follows.

\begin{theorem}[Informal] There is an (almost)\footnote{See Remark \ref{rem:epsilon}.} 
 polynomial time algorithm for approximately sampling from the Gibbs distribution in $u$-capacitated $k$-uniform congestion games assuming that $1 \leq k_{\max}u_{\max} = o\left(U^{1/4}\right)$ when $n \rightarrow \infty$, where $k_{\max} = \max_i k_i$, $u_{\max} = \max_j u_j$ and $U = \sum_j u_j$. 
\label{thm:cap_informal}
\end{theorem}

The proof of Theorem \ref{thm:cap_informal} reveals an interesting connection between M-convexity and asymptotic enumeration formulas, that might be of independent interest.\\ 

\noindent To the best of our knowledge, ours are the first sampling results for the Gibbs distribution in congestion games, beyond the well-studied case of Glauber dynamics for the Ising model (see Section \ref{sec:rel_work}). Given the extensive attention that logit dynamics has received in various communities, we believe this to be an interesting line of work, at the intersection of algorithmic game theory, combinatorics and approximate sampling, to pursue further. 

In particular, for special cases of congestion games with a positive answer to questions $(1)$ and $(2)$ as in the introduction, do there also exist positive answers for their dynamic analogues? As a concrete open question, we ask whether it is always possible to efficiently sample from the Gibbs distribution in general base matroid congestion games \cite{Ackermann2008}. If true, this would provide an interesting (qualitative) game-theoretical generalization of the sampling result of Anari et al. \cite{SLC2}.

\subsection{Further related work}\label{sec:rel_work}
For an exposition of the notion of logit equilibrium, and more related work, we refer to the survey article of Ferraioli \cite{Ferraioli2013} and references therein. There are also various results addressing the inefficiency of `long-term' equilibria in the context of logit dynamics, see, e.g., the works of Asadpour and Saberi \cite{Asadpour2009}, Mamageishvili and Penna \cite{Mamageishvili2016} and Penna \cite{Penna2018}.
 
A special case of potential games for which logit-like dynamics have been studied extensively is the Glauber dynamics for the Ising model, which, in game-theoretical terms, can be seen as logit dynamics in so-called \emph{max cut games}, see, e.g., \cite{GM09}.\footnote{There exist many generalizations of max cut games, see, e.g., \cite{KS2019} and references therein, for which it might also be interesting to study the logit dynamics.} Here, we are given a graph $G = (V,E)$ of which its nodes $V$ are players that all have strategy set $\{-1,+1\}$. The potential function in this case is given by $\Phi(s) = \sum_{\{i,j\} \in E} s_is_j$ for a strategy profile $s = (s_1,\dots,s_n)$. Whether or not the logit dynamics are rapidly mixing in this case depends on the parameter $T \geq 0$ and graph topology $G$, see, e.g., the work of Levin, Luczak and Peres \cite{Levin2010} and references therein. Jerrum and Sinclair \cite{Jerrum1993} show that, nevertheless, there exists a polynomial time algorithm to sample from the Gibbs distribution for any parameter $T \geq 0$ and graph topology. The first question posed in the introduction essentially aims at investigating to what extent a similar result is possible for (special classes of) congestion games.

\section{Preliminaries}\label{sec:pre}
In this section we will give all the necessary preliminaries regarding resource allocation (or congestion) games, strongly log-concave polynomials and the relevant Markov chain notions and results. We start with some general notation.

All logarithms in this work have Euler's number $e$ as their base, unless specified otherwise. For $k \in \Z_{>0}$, we write $[k] = \{1,\dots,k\}$. For two vectors $x,y \in \Z^n$, we write $x \leq y$ if $x_i \leq y_i$ for $i = 1,\dots,n$, and $x < y$ if strict inequality holds for at least one $i$. Furthermore, with $|x| = \sum_{i=1}^n |x_i|$ we denote the modulus of $x$. We use $(e_i)_{i = 1,\dots,n}$ to denote the standard basis of $\R^n$, i.e., $e_i(k) = 1$ if $k = i$ and $e_i(k) = 0$ otherwise. 

\subsection{Congestion games}\label{sec:pre_cong}
A \emph{capacitated congestion game} $\Gamma$ is given by a tuple $(N,E,(\mathcal{S}_i)_{i\in N},(c_e)_{e\in E},(u_e)_{e \in E})$, where $N = [n]$ is a finite set of players, $E = [m]$ a finite set of resources (or facilities), $\mathcal{S}_i \subseteq 2^{E}$ is a set of strategies of player $i \in N$, and $c_e : \Z_{\geq 0} \rightarrow \Q$ the cost function of resource $e \in E$ that satisfies $c_e(x) = W$ whenever $x > u_e$ for $e \in E$ with $W$ a (sufficiently) large number.\footnote{Think of $W$ as being $\infty$.} Unless stated otherwise, the cost functions are assumed to be non-negative and non-decreasing. Finally, $u_e$ is non-negative integer modeling the capacity on resource $e \in E$. If $u_e = n$ for every resource $e \in E$, we simply call $\Gamma$ a \emph{congestion game}.\footnote{Extension parallel and capacitated uniform congestion games are defined in Sections \ref{sec:ep} and \ref{sec:cap}, respectively.} For a strategy profile $s = (s_1,\dots,s_n) \in \times_i \mathcal{S}_i = \mathcal{S}$, we define $\ell_e(s)$ as the number of players using resource $e$, i.e., $\ell_e(s) = |\{i \in N : e \in s_i\}|$. A game is called \emph{symmetric} if $\mathcal{S}_i = \mathcal{S}_j$ for all $i,j \in [n]$. We then write $\mathcal{P}$ to denote the common strategy set of all players.

We call $\ell(s) = (\ell_e(s))_{e \in E}$ the \emph{resource load profile corresponding to strategy profile $s$}. We say that a strategy $s \in \mathcal{S}$ is \emph{feasible} if $\ell_e(s) \leq u_e$ for every $e \in E$ and write $\mathcal{S}_f$ to denote the set of all feasible strategy profiles.\footnote{We only consider games in which the set of feasible strategy profiles is non-empty.}  
More generally, we say that $y \in \mathbb{N}^m$ is a \emph{(feasible) resource  load profile} for $(N,E,(\mathcal{S}_i)_{i\in N},(u_e)_{e \in E})$ if there is some (feasible) strategy profile $s$ such that $y= \ell(s)$. We write $\mathcal{S}(y)$ for the set of strategy profiles $s \in \times_i \mathcal{S}_i$ whose resource load profile is $y$. 

Similarly, for symmetric congestion games we define the notion of a \emph{strategy load profile} that models how many players are using a strategy $p \in \mathcal{P}$ in a given strategy profile $s \in \mathcal{P}^n$. More precisely, given a strategy profile $s \in \mathcal{P}^n$, we define $z_p(s) = |\{ i \in N : s_i = p\}|$ as the number of players choosing strategy $p \in \mathcal{P}$ in strategy profile $s$. The vector $z(s) = (z_p(s))_{p \in \mathcal{P}}$ is called the strategy load profile of $s$. Similarly as for resource load profiles, we define (with a slight abuse of notation) $\mathcal{S}(x)$ as the set of strategy profiles $s$ for which $x =  z(s)$.\\

\noindent The cost of player $i \in N$ under a strategy profile $s = (s_1,\dots,s_n) \in \times_i \mathcal{S}_i$ is given by 
$
C_i(s) = \sum_{e \in s_i} c_e(\ell_e(s)).
$
A strategy profile $s \in \mathcal{S}$ is called a \emph{(pure) Nash equilibrium} if for every $i \in N$ and every $s_i' \in \mathcal{S}_i$ it holds that $C_i(s) \leq C_i(s_i',s_{-i})$, where $(s_i',s_{-i})$ denotes the strategy profile in which player $i$ plays $s_i'$ and every other player $j \neq i$ plays $s_j$. We write $\text{NE}(\Gamma)$ to denote the set of all pure Nash equilibria of $\Gamma$.

We say that $\Phi : \times_i \mathcal{S}_i \rightarrow \R$ is an \textit{exact potential function} for a congestion game $\Gamma$ if for every strategy profile $s \in \times_i \mathcal{S}_i$, for every player $i \in N$ and every unilateral deviation $s'_i \in \mathcal{S}_i$ of $i$ it holds: 
$
\Phi(s) - \Phi(s_{-i},s_i') = C_i(s) - C_i(s_{-i},s_i').
$
Rosenthal \cite{Rosenthal1973} shows that 
\begin{equation}\label{eq:rosenthal}
\Phi(s) = \sum_{e \in E} \sum_{k=1}^{\ell_e(s)} c_e(k)
\end{equation}
is an exact potential function for any congestion game). 
Subsequently, we refer to this potential function as \emph{Rosenthal's potential}.

 A function $\phi : \{0,\dots,n\} \rightarrow \R$, is called \emph{convex} if $\phi(i) - \phi(i-1) \leq \phi(i+1) - \phi(i)$ for all $i = 1,\dots,n-1$. A function $\psi : \{0,\dots,n\}^m \rightarrow \R$ is called \emph{separable convex} if it is of the form $(x_1,\dots,x_m) \mapsto \sum_{j=1}^m \psi_j(x_j)$ where the $\psi_j$, given by $x_j \mapsto \psi_j(x_j)$, are convex. We say that $\phi$ is \emph{concave} if $-\phi$ is convex, and, similarly, $\psi$ is \emph{separable concave} if $-\psi$ is separable convex. A simple, but important, observation that we will use in this work is the fact that Rosenthal's potential is a separable convex function when seen as a function from resource load profiles to the reals. That is, the function $\bar{\Phi}: \Z_{\geq 0}^m \rightarrow \R$, given by 
\begin{equation}\label{eq:rosenthal_resource}
\bar{\Phi}(\alpha) = \sum_{e \in E} \sum_{k=1}^{\alpha_e} c_e(k)
\end{equation}
for $\alpha \in \Z_{\geq 0}^m$ is separable convex.
\begin{proposition}\label{prop:ros_sep}
If the cost functions $(c_e)_{e \in E}$ are non-decreasing, then Rosenthal's potential $\bar{\Phi}$ is a separable convex function.
\end{proposition}

\paragraph{The Gibbs distribution and logit dynamics.}\label{sec:pre_gibbs}
The Gibbs distribution $\pi : \mathcal{S} \rightarrow \R_{\geq 0}$ over the strategy profiles of a congestion game $\Gamma$ is given by 
$$
\pi(s) = \frac{e^{-T \Phi(s)}}{Z}
$$
where $\Phi$ is Rosenthal's potential, $T \geq 0$ a so-called temperature parameter, and $Z$ the normalizing constant (or partition function)
$$
Z = \sum_{t \in \mathcal{S}_f} e^{-T \Phi(t)}.
$$ 

The \emph{logit dynamics Markov chain}\footnote{Formal Markov chain definitions are given in Section \ref{sec:markov}.} with current state $s \in \mathcal{S}$ proceeds by:
\begin{itemize}
\item Select a player $i \in N$ uniformly at random.
\item For $s_i' \in \mathcal{S}_i$, transition to $(s_i',s_{-i})$ with probability 
$e^{-T \Phi(s_i',s_{-i})}/Z_i'$ with normalizing constant
$$
Z_i' = \sum_{r \in \mathcal{S}_i} e^{-T \Phi(r,s_{-i})}.
$$
\end{itemize}

It is a standard fact, which can be shown by using \eqref{eq:exact}, that this Markov chain is reversible with respect to the  Gibbs distribution.

\subsection{Matroids and $M$-concavity}
\label{sec:matroid}
Let $E = [n]$ be a finite set called the \emph{ground set} and $\mathcal{I} \subseteq 2^E = \{X : X \subseteq E\}$ a collection of subsets of $E$ (called \emph{independent sets}). The pair $\mathcal{N} = (E,\mathcal{I})$ is a \emph{matroid} if  i) $\emptyset \in \mathcal{I}$; ii) $A \in \mathcal{I}$ and $B \subseteq A$, then $B \in \mathcal{I}$; iii) $A,B \in \mathcal{I}$ and $|A| > |B|$, then there exists an $a \in A \setminus B$ such that $B + a \in \mathcal{I}$.\footnote{For $A,B,C \subseteq E$ with $|B|, |C| \leq 1$, we use the notation $A - B + C$ to denote the set $(A \setminus B) \cup C$.}
An independent set $B \in \mathcal{I}$ of maximum size is called a \emph{basis}. We use $\mathcal{B}$ to denote the set of all bases of $\mathcal{N}$. The set of bases $\mathcal{B}$ satisfies the so-called \emph{base-exchange property}: If $B,B' \in \mathcal{B}$ and $e \in B \setminus B'$, then there exists an $e' \in B' \setminus B$ such that $B + e' - e \in B$. It safisfies the \emph{strong base-exchange property} if both $B + e' - e,\ B' - e' + e \in B$. The \emph{rank} of a matroid is the common cardinality $r$ of all bases in $\mathcal{B}$. The $\ell$-truncation $\mathcal{M}_{\ell} = (E,\mathcal{I}_\ell)$ of a matroid $\mathcal{M}$ is the matroid with $A \in\mathcal{I}_\ell$ if and only if $A \in \mathcal{I}$ and $|A| \leq \ell$.
The \emph{partition matroid} is given by a disjoint partition $E = E_1 \cup \dots \cup E_q$ of the ground set $E$, and upper bounds $u_i$ for $i = 1,\dots,q$. A set $A \subseteq E$ is independent if and only if $|A \cap E_i| \leq u_i$ for all $i = 1,\dots,q$. The \emph{$k$-uniform matroid} is the matroid in which $A \subseteq E$ is independent if and only if $|A| \leq k$. 

A \emph{discrete polymatroid} (which can be seen as a multi-set generalization of a matroid) is a finite set of vectors $R \subset \Z_{\geq 0}^n$ with the properties that i) $\mathbf{0} \in R$; ii) if $y \in R$ and $x \leq y$, then $x \in R$; and iii) if $x,y \in R$ with $|y|> |x|$, then there is a vector $w \in R$ such that $x < w < \max\{u,v\}$ (where the maximum is taken coordinate-wise). The set of bases $\mathcal{B}_R$ is given by all maximal vectors in $R$ that have a common modulus $r$.  A polymatroid satsifies the \emph{base-exchange property}: if $x,y \in \mathcal{B}_R$ and $x_i > y_i$, then there exists an index $j$ with $y_j > x_j$ and $x - e_i + e_j \in R$.

Finally, as a generalization of discrete polymatroids, we describe the notion of \emph{$M$-convexity} for functions \cite{Murota1998,Murota2009}. As we will mostly work with its negated counterpart of \emph{$M$-concavity}, we will describe this first. Let $\nu : \Z_{\geq 0}^n \rightarrow \R \cup \{-\infty\}$ be a function.\footnote{In particular, we set $\log(0) = -\infty$.} The \emph{effective domain} of $\nu$ is given by 
$$
\text{dom}(\nu) = \{\alpha \in \Z_{\geq 0}^n : v(\alpha) > - \infty\}.
$$
The function $\nu$ is called \emph{$M^\sharp$-concave} if it satisfies the \emph{(symmetric) exchange property}: For any $\alpha, \beta \in \text{dom}(\nu)$ and any $i \in [n]$ satisfying $\alpha_i > \beta_i$, there exists a $j \in [n]$ such that $\alpha_j < \beta_j$ and
\begin{equation}\label{eq:M-concave}
\nu(\alpha) + \nu(\beta) \leq \nu(\alpha - e_i + e_j) + \nu(\beta + e_i - e_j).
\end{equation}
It is well known that a separable concave function is $M^\sharp$-concave \cite{Murota2003}. The function $\nu$ is called \emph{$M$-concave} if it is $M^{\sharp}$-concave and, in addition,  there is an $r \in \Z_{\geq 0}$ such that $\text{dom}(\nu) \subseteq \{\alpha: \sum_i \alpha_i = r\}$. 
A function $\nu : \Z_{\geq 0}^n \rightarrow \R \cup \{\infty\}$ is called \emph{$M$-convex} if $-\nu$ is $M$-concave.

\subsection{Strongly log-concave polynomials}
We consider polynomials $p \in \R[x_1,\dots,x_n]$ with non-negative coefficients. For a vector $\beta = (\beta_1,\dots,\beta_n) \in \Z_{\geq 0}^n$, we write 
$$
\partial^\beta = \prod_{i=1}^n \partial_{x_i}^{\beta_i}
$$
to denote the partial differential operator that differentiates a function $\beta_i$ times with respect to $x_i$ for $i = 1,\dots, n$.  
For $\alpha \in \Z_{\geq 0}^n$, we write $x^\alpha$ to denote $\prod_{i =1}^n x_i^{\alpha_i}$. Furthermore, we write $\alpha! = \prod_i \alpha_i!$, and for $\alpha, \kappa \in \Z_{\geq 0}^n$ with $\alpha_i \leq \kappa_i$ for all $i$, we write
$$
\binom{\kappa}{\alpha} = \prod_{i = 1}^n \binom{\kappa_i}{\alpha_i}.
$$
For a constant $t \in \Z_{\geq 0}$ with $t \geq \max_i \alpha_i$, we write $\binom{t}{\alpha} = \prod_{i = 1}^n \binom{t}{\alpha_i}$.
Let $\kappa \in \Z_{\geq 0}^n$ and $K = \times_i \{0,\dots,\kappa_i\}$. Let $w : K \rightarrow \R_{\geq 0}$ be a weight function. The \emph{generating polynomial} of $w$ is given by 
$$
g_{\kappa}(x) = \sum_{\alpha \in K} w(\alpha)x^{\alpha}.
$$
The \emph{support} of $g_{\kappa}$ is the set $\text{supp}(g_{\kappa}) = \{\alpha \in K : w(\alpha) > 0\}$.
The generating polynomial $g$ is called $d$-homogeneous if $|\alpha| = \sum_i \alpha_i = d$ for all $\alpha \in \{0,\dots,k\}^m$ with $w_{\alpha} > 0$. It is called \emph{multi-affine} if every variable $x_i$ appears with at most multiplicity one in every monomial of $p$. For example, $q(x_1,x_2) = x_1x_2$ is multi-affine, but $r(x_1,x_2) = x_1^2 + x_1x_2$ is not, as the multiplicity of $x_1$ in the first monomial is two.
Finally, the \emph{elementary symmetric polynomial of degree $d$}, for $\kappa = (1,1,\dots,1)$, is given by 
$$
h_{\kappa}(x) = \sum_{\alpha \in \{0,1\}^n : |\alpha| = d} x^\alpha.
$$

\begin{definition}[Strong log-concavity \cite{G2009}]
A polynomial $p \in \R[x_1,\dots,x_n]$ with non-negative coefficients is called \emph{log-concave} on a subset $S \subseteq \R^n_{\geq 0}$ if its Hessian $\triangledown^2 \log(p)$ is negative semidefinite on $S$. A polynomial $p$ is called \emph{strongly log-concave (SLC)} on $S$ if for any $\beta \in \Z_{\geq 0}^n$, we have that $\partial^\beta p$ is log-concave.
\end{definition}
\noindent For convenience, the zero polynomial is defined to be strongly log-concave always. It is interesting to note that if a $d$-homogeneous multi-affine polynomial $p$ is SLC, then the support of $p$ must form the collection of bases of a matroid, and, more general, if a (not multi-affine) homogeneous polynomial is SLC, its support forms an $M$-convex set \cite{BH2019}.

Finally, if the generating polynomial $g_{\kappa}$ is strongly log-concave, then the probability distribution $\pi(\alpha) \propto w(\alpha)$ is called strongly log-concave.

\begin{remark}
The definition of strong log-concavity is (in this work) is not really needed, but included for completeness. In our proofs we essentially only rely on properties of SLC polynomials from the literature that are reviewed below. 
For homogeneous generating polynomials the notion of strong log-concavity is equivalent to that of a polynomial being \emph{Lorentzian} \cite{BH2019}, or \emph{completely log-concave} \cite{ALOV2018}. These equivalences are shown in \cite{BH2018}. 
\end{remark}

\noindent We next state all properties of SLC polynomials that will be used in this work.
First of all, it is easy to check that the SLC property is preserved under multiplication with a non-negative scalar, which we state below for sake of reference.

\begin{proposition}[Br\"and\'en and Huh \cite{BH2019}]\label{prop:scalar}
If $p \in \R[x_1,\dots,x_n]$ is SLC and $\gamma \in \R_{\geq 0}$, then $\gamma  p$ is SLC.
\end{proposition}

We continue with the \emph{polarization operator} defined by Br\"and\'en and Huh \cite{BH2019}. 
The polarization operator introduces auxiliary variables in order to turn $p$ into a multi-affine polynomial over a larger set of variables.\footnote{We elaborate on polarization in Section \ref{sec:polymatroid}.} 
Formally, following \cite{BH2019}, for $\kappa \in \Z_{\geq 0}^n$ let
$$
\R_{\kappa}[x_i] = \Big\{\text{polynomials in } \R[x_i]_{1 \leq i \leq n} \text{ of degree at most } \kappa_i \text{ in } x_i \text{ for every } i\Big\},
$$
and
$
\R^a_{\kappa}[x_{ij}] = \left\{\text{multi-affine polynomials in } \R[x_{ij}]_{1 \leq i \leq n, 1 \leq j \leq \kappa_i}\right\}.
$
The polarization operator $\Pi_\kappa : \R_{\kappa}[x_i] \rightarrow \R^a_{\kappa}[x_{ij}]$ replaces every factor $x^\alpha$ by 
$$
\frac{1}{\binom{\kappa}{\alpha}}\prod_{i = 1}^n \Big( \text{elementary symmetric polynomial of degree } \alpha_i \text{ in the variables } \{x_{ij}\}_{1 \leq j \leq \kappa_i} \Big). 
$$

\begin{proposition}[Br\"and\'en and Huh \cite{BH2019}]\label{prop:homo}
If $p$ is $d$-homogeneous and SLC over $\R_{\kappa}[x_i]$, then $\Pi_\kappa(p)$ is $d$-homogeneous  and SLC over $\R^a_{\kappa}[x_{ij}]$.
\end{proposition}

We conclude this section by stating a large class of homogeneous polynomials that are known to be strongly log-concave.

\begin{proposition}[Br\"and\'en and Huh \cite{BH2019}]\label{prop:sufficient}
For $\kappa \in \Z_{\geq 0}^n$ and $w : K \rightarrow \R_{\geq 0}$ a non-negative weight function, consider
\begin{equation}\label{eq:standard_lorentzian}
g_{\kappa}(x) = \sum_{\alpha \in K} \frac{w(\alpha)}{\alpha!}x^{\alpha},
\end{equation}
and assume that $g_{\kappa}$ is $d$-homogeneous.
Let $\nu : \Z_{\geq 0}^n \rightarrow \R \cup \{-\infty\}$ be defined by
$
\nu(\alpha) = \log(w(\alpha))
$
for $\alpha \in K$ and $\nu(\alpha) = -\infty$ otherwise. If $\nu$ is $M$-concave, then $g_{\kappa}$ is SLC.
\end{proposition}

\subsection{Markov chains}\label{sec:markov}
Let $\mathcal{M} = (\Omega,P)$ be an aperiodic, irreducible and time-reversible Markov chain with state space $\Omega$, transition matrix $P$, and stationary distribution $\pi$. Reversibility means that
$
\pi(x)P(x,y) = \pi(y)P(y,x)
$
for any $x, y \in \Omega$. 
We write $P^t(x,\cdot)$ for the distribution over $\Omega$ at time step $t$ with initial state $x \in \Omega$. The \emph{total variation distance} $d_{TV}(\pi,\sigma)$ of two distributions $\pi$ and $\sigma$ over $\Omega$ is defined as 
$
d_{TV}(\pi,\sigma) = \max_{S \subseteq \Omega} |\pi(S) - \sigma(S)| = \frac{1}{2}\sum_{x \in \Omega} |\pi(x) - \sigma(x)|,
$
where for a distribution $\sigma$ over $\Omega$, we write $\sigma(S) = \sum_{x \in S} \sigma(x)$. 
We say that two distributions $\pi$ and $\sigma$ are \emph{$\epsilon$-close} if $d_{TV}(\pi,\sigma) \leq \epsilon$.
The total variation distance of the distribution $P^t(x,\cdot)$ from $\pi$ at time $t$ with initial state $x$ is denoted by $\Delta_x(t)$. 
The \emph{mixing time} of $\mathcal{M}$ with initial state $x \in \Omega$ is
$
\tau_x(\epsilon) = \min\{ t : \Delta_x(t') \leq \epsilon \text{ for all } t' \geq t\}.
$
Informally, $\tau_x(\epsilon)$  is the number of steps until the Markov chain is $\epsilon$-close to its stationary distribution, given that it is starting in $x$.  A treaty of some more advanced Markov chain notions is given in Appendix \ref{app:markov}, including the definition of the \emph{modified log-Sobolev constant $\rho = \rho(P)$} which can be used to bound the mixing time of a Markov chain.\footnote{The definition of $\rho$ is deferred to Appendix \ref{app:markov} as we actually do not need it in this work; we only rely on lower bounds on this constant obtained by other authors.} In particular, it holds that
\begin{equation}
\label{eq:pre_modsob}
 \tau_x(\epsilon) \leq  \frac{1}{\rho(P)}\left( \log \log \pi(x)^{-1} + \log\left(\frac{1}{2\epsilon^2} \right)\right).
\end{equation}

\paragraph{Markov chain decomposition.}\label{sec:markov_decomp}
Let $\Omega = \Omega_1 \cup \dots \cup \Omega_m$ be a disjoint partition of the state space $\Omega$. Following \cite{Martin2000}, consider
$
\bar{\pi}(i) = \pi(\Omega_i) = \sum_{x \in \Omega} \pi(x)
$
and let $\bar{P} : [m] \times [m] \rightarrow [0,1]$ be defined by
$$
\bar{P}(i,j) = \bar{\pi}(i)^{-1}\sum_{x \in \Omega_i,\ y \in \Omega_j} \pi(x)P(x,y).
$$
The Markov chain on $[m]$ with transition matrix $\bar{P}$ is called the \emph{projection chain} on the partition $\{\Omega_i\}_{i = 1,\dots,m}$. It is time-reversible with respect to the distribution $\bar{\pi}$ over $[m]$. For $i \in [m]$ the \emph{restriction chain} on $\Omega_i$ has transition matrix $P_i: \Omega_i \times \Omega_i \rightarrow [0,1]$ given by
$$
P_i(x,y) = \left\{ \begin{array}{ll}
P(x,y) & \text{ if } x \neq y,\\
1 - \sum_{z \in \Omega_i \setminus \{x\}} P(x,z) & \text{ if } x = y.
\end{array}\right.
$$
Its stationary distribution $\pi_i$ is given by $\pi_i(x) = \pi(x)/\bar{\pi}(i)$ for $x \in \Omega_i$. In Appendix \ref{app:markov_decomp} we give a Markov chain decomposition result based on the modified log-Sobolev constant $\rho$.

\paragraph{Base-exchange Markov chain.}
Let $\mathcal{N}$ be a matroid and let $\pi$ be an SLC probability distribution over the set of bases $\mathcal{B}$ given by $\pi(\alpha) \propto w(\alpha)$ for some non-negative weight function $w: K \rightarrow \R$. Here $\pi(\alpha) \propto w(\alpha)$ means that $\pi(\alpha) = w(\alpha)/(\sum_{\alpha \in K} w(\alpha))$.
The \emph{base-exchange Markov chain on $\mathcal{B}$} is defined by the following transitions, where $B \in \mathcal{B}$ is the current state of the Markov chain:
\begin{itemize}
\item Select an element $e \in B$ uniformly at random and remove it.
\item Pick a base $B' \in \mathcal{B}$ with $B' \supset B - e$ with probability $\propto w(B')$ among all such bases $B'$.
\end{itemize}
It is not hard to see, using the base-exchange property, that this procedure defines an ergodic, time-reversible Markov chain with stationary distribution $\pi$. Anari et al. \cite{ALOV2018} show that this chain is rapidly mixing for any matroid $\mathcal{N}$. In particular, in a recent follow-up work, they give a (tight) mixing time bound \cite{SLC4}.

\begin{theorem}[Anari et al. \cite{SLC4}]\label{thm:base_walk}
Let $\mathcal{N}$ be matroid of rank $r$, and let $\pi$ be an SLC probability distribution over the set of bases $\mathcal{B}$ given by $\pi(\alpha) \propto w(\alpha)$ for some weight function $w: K \rightarrow \R_{\geq 0}$.  Then the mixing time of the base-exchange random walk, satisfies
$
\tau(\epsilon) \leq O(r \log(r/\epsilon)).
$
\end{theorem}
We note that the mixing time is independent of the size $n$ of the ground set $E$ of the matroid $\mathcal{N}$, as well as the stationary distribution $\pi$.  In this work Theorem \ref{thm:base_walk} is essentially only applied to partition and uniform matroids. 
Furthermore, Cryan, Guo and Mousa \cite{Cryan2019} show that the modified log-Sobolev of the base-exchange random walk satisfies $\rho \geq 1/r$, where $r$ is the rank of the matroid.

\subsection{Sampling algorithms}\label{sec:sampling}
Consider a class of (capacitated) congestion games $\Gamma = (N,E,(\mathcal{S}_i)_{i\in N},(c_e)_{e\in E},(u_e)_{e \in E})$ with $n$ players and $m$ resources, and cost functions $c_e : \Z_{\geq 0} \rightarrow \Q$ for $e \in E$. Let $w : \mathcal{S} \rightarrow \Q_{\geq 0}$ be a weight function. In this work, an algorithm for sampling $s \in \mathcal{S}$ according to a distribution $\epsilon$-close to $\pi$ with $\pi(s) \propto w(s)$ is said to run in \emph{(randomized) polynomial time} if the number of arithmetic operations can be upper bounded by a polynomial in $n,\ m,\footnote{Assuming that the strategy sets can be described in a ``compact'' form, like in the case of (symmetric) network congestion games. Alternatively, one could replace $m$ by $\max_i |\Si|$ in case the strategy sets are assumed to be given explicitly.}\ \log(1/\epsilon),\ \max_{e,j} \log(c_e(j))$ and $\max_s \log(w(s))$.  
The generation of a uniform random $0/1$ bit is considered to be one arithmetic operation.

\begin{remark}[Real numbers]
In this work we use Markov chains whose transition probabilities are, in general, not rational numbers (in particular for the Gibbs distribution). Whenever we use real numbers, it is implicitly assumed that we use sufficiently accurate approximations to these numbers. All our results remain valid when (real-valued) transition probabilities are replaced by sufficiently accurate rational approximations.  
We note that, roughly speaking, whenever we want to generate Markov chain transitions with probabilities proportional to $e^{-T\Phi(s)}$ for $s \in \mathcal{S}$, our algorithms run in pseudo-polynomial time in terms of the values of the cost functions of the congestion game under consideration.
\end{remark}

All our results algorithmic results are based on running Markov chains for a sufficiently long time. 
We will usually write our running time bounds as the product of two factors: The number of steps that we need to run the Markov chain (before it's close to stationarity), and the complexity of implementing one such step. In particular, in all cases, the transitions probabilities of one step are determined by a sequence of rational numbers $a = (a_1,\dots,a_z)$ and $q = (q_1,\dots,q_z)$, and we want to sample an index $i \in [z]$ with probability
\begin{equation}\label{eq:trans_real}
\frac{q_ie^{a_i}}{\sum_{i} q_ie^{a_i}}.
\end{equation}
Here $z$, as well as the encoding size of the $q_i$, will be $\text{poly}(n,m)$. We will refer to $C = C(n,m,a)$ as the computational complexity of sampling an index $i$ according to (approximations of) the above probabilities in order not to overload our theorem statements. We say that probabilities of the form \eqref{eq:trans_real} are \emph{suitable}.

\subsection{Bipartite graphs}
An (undirected) bipartite graph $G = (A \cup B, F)$ is given by two disjoint sets of nodes $A = \{a_1,\dots,a_n\}$ and $B = \{b_1,\dots,b_m\}$ with $F \subseteq \{\{a,b\} : a \in A, b \in B\}$. We say that $G$ has degree sequence $(\mathbf{x},\mathbf{y})$ if $d(a_i) = x_i$ for $i = 1,\dots,n$ and $d(b_i) = y_j$ for $j = 1,\dots,m$, where $d(v)$ denotes the degree of node $v$ in $G$. We write $\mathcal{G}(\mathbf{x},\mathbf{y})$ for the set of all bipartite graphs on $A \cup B$ with degree sequence $(\mathbf{x},\mathbf{y})$. 

\section{General approach}\label{sec:overview}
In Section \ref{sec:pre_cong}, we gave two possible definitions for the load profile of a strategy profile $s = (s_1,\dots,s_n) \in \times_i \Si$. For general congestion games, we defined the resource load profile $\ell(s) = (\ell_e(s))$ that keeps track of how many players use a particular resource $e$ in $s$. For symmetric congestion games, we may in addition consider the strategy load profile $z(s) = (z_t(s))_{t \in \mathcal{S}_0}$ that keeps track of how many player use a particular strategy $t$ from the common strategy set $\mathcal{P}$.

A general approach for sampling a strategy profile according to the Gibbs distribution in Section \ref{sec:ep_gibbs} and \ref{sec:cap_gibbs} will be to first sample a (resource or strategy) load profile $\alpha$ according to approximately the right probability, and then sample a strategy profile $s \in \mathcal{S}(\alpha)$ uniformly at random. Remember that the set $\mathcal{S}(\alpha)$ was used to denote all strategy profiles with the given (resource or strategy) load profile. The approach is summarized in Algorithm \ref{alg:ep}, in which we essentially give the formulation for resource load profiles, since we use Rosenthal's potential $\bar{\Phi} : \times_i\Si \rightarrow \R$ as given in \eqref{eq:rosenthal_resource}. The formulation for strategy load profiles is exactly the same, with $\bar{\Phi}$ replaced by $\Phi$ as given in \eqref{eq:rosenthal}. We note that, in order to sample a strategy load profile $s \in \mathcal{S}(\alpha)$ uniformly at random in symmetric games, it suffices to generate a random permutation of the players in $N = \{1,\dots,n\}$. 

Sampling a load profile with the correct probability in our applications corresponds to sampling a base of a discrete polymatroid according to a strongly log-concave distribution. In order to do this, we present a reduction of this problem to that of sampling a base of a matroid according to a strongly log-concave distribution in Section \ref{sec:polymatroid} (after which we can rely on Theorem \ref{thm:base_walk}).

\vspace*{-0.2cm}
\IncMargin{1em}
\begin{algorithm}
\SetKwData{Left}{left}\SetKwData{This}{this}\SetKwData{Up}{up}
\SetKwFunction{Union}{Union}\SetKwFunction{FindCompress}{FindCompress}
\SetKwInOut{Input}{Input}\SetKwInOut{Output}{Output}
\Input{Congestion game $\Gamma$, temperature $T \geq 0$ and $\epsilon \geq 0$.}
\Output{Strategy profile $s \in \mathcal{S}$ according to distribution $\bar{\pi}$ that is $\epsilon$-close to Gibbs distribution $\pi$ at temperature $T$.}\medskip

\textbf{Step I:} Sample load profile $\alpha$ according to a distribution $\sigma'$ that is $\epsilon$-close to $\pi'$ given by
$$
\pi'(\alpha) = |\mathcal{S}(\alpha)| e^{-T \bar{\Phi}(\alpha)}.
$$\\
\textbf{Step II:} Sample strategy profile $s \in \mathcal{S}(\alpha)$ (approximate) uniformly at random.
\caption{Gibbs sampler for congestion game $\Gamma$.}
\label{alg:ep}
\end{algorithm}
\DecMargin{1em}
\vspace*{-0.4cm}

\subsection{Sampling bases of discrete polymatroids}\label{sec:polymatroid}
In this section we describe how to generate a discrete polymatroid base, according to a strongly log-concave distribution over the set of all polymatroid bases, by reducing it to the problem of generating a base of a matroid.\footnote{This is a special case of sampling an element from an M-convex set under a strongly log-concave distribution.} 
This follows more or less directly from Theorem \ref{thm:base_walk} by using the notion of polarization. Polarization can be seen as a functional version of the classical reduction from discrete polymatroids to matroids, as given by Helgason \cite{Helgason1974}.\footnote{See also Chapter 44.6b in \cite{Schrijver2003} for this reduction.}

For a polymatroid $R \subset \Z_{\geq 0}^n$, consider a $d$-homogeneous  strongly log-concave  polynomial
$
g_R(x_1,\dots,x_n) = \sum_{\alpha \in \mathcal{B}_R} w(\alpha) x^{\alpha}.
$
with positive coefficients and support the set of bases $\mathcal{B}_R$.
Consider the matroid $\mathcal{N}_R = (E,\mathcal{I})$ on ground set $E = \{(i,j) : 1 \leq i \leq n,\ 1 \leq j \leq d\}$, where $I \in \mathcal{I}$ if and only if the vector $\alpha(I) \in \Z_{\geq 0}^n$ given by $\alpha_e = |\{f : (e,f) \in I\}|$ satisfies $\alpha(I) \in R$. The fact that $\mathcal{N}_R$ is indeed a matroid follows directly from the fact that $R$ is a polymatroid. Note that, for a given $\alpha \in R$, we have
\begin{equation}
\label{eq:scaling}
|\{I : \alpha(I) = \alpha\}| = \binom{d}{\alpha}
\end{equation}
We slightly abuse notation here and write $d = (d,\dots,d)$ for the all $d$-vector in $\Z^n$.

Then, with $\Pi(\mathcal{B}_R)$ the set of bases of $\mathcal{N}_R$, the polarization $\Pi(g)$ of $g$ can be written as 
\begin{align*}
\Pi(g)(y_{11},\dots, y_{1d},\dots,y_{n1},\dots,y_{nd}) &= \sum_{B \in \Pi(\mathcal{B}_R)} \binom{d}{\alpha(B)}^{-1} w(\alpha(B))\cdot  y^{B} \\
&= \sum_{B \in \Pi(\mathcal{B}_R)} w_\Pi(\alpha(B))\cdot  y^{B}
\end{align*}
where, for $B \in  \Pi(\mathcal{B}_R)$, we define 
\begin{equation}
\label{eq:poly_to_matroid}
w_\Pi(\alpha(B)) = \binom{d}{\alpha(B)}^{-1} w(\alpha(B)).
\end{equation}
Polarization should be interpreted as spreading out the weight $w(\alpha)$ for $\alpha \in R$ equally over all bases $B \in \{A : \alpha(A) = \alpha\} \subseteq \Pi(\mathcal{B}_R)$. Proposition \ref{prop:homo} implies that $\Pi(g)$ is also strongly log-concave. 

\begin{example}
Let $p(x_1,x_2) = x_1^2x_2 + x_1x_2$, so that $\text{supp}(p) = \{(2,0),(1,1)\}$, and take $d = (2,2)$. Then 
\begin{align*}
\Pi_d(p) & =\frac{1}{2}x_{11}x_{12}(x_{21} + x_{22}) + \frac{1}{4}(x_{11} + x_{12})(x_{21} + x_{22}) \\
&= \frac{1}{2}x_{11}x_{12}x_{21} + \frac{1}{2}x_{11}x_{12}x_{22} + \frac{1}{4}x_{11}x_{21} + \frac{1}{4}x_{11}x_{22} + \frac{1}{4}x_{12}x_{21} + \frac{1}{4}x_{12}x_{22}.
\end{align*}
Note that, looking at the support of $p$, we have $\binom{d}{\alpha} = 2$ monomials corresponding to $\alpha = (2,0)$ and $\binom{d}{\alpha} = 4$ monomials corresponding to $\alpha = (1,1)$.
\end{example}

Corollary \ref{cor:base_walk} below now follows directly from Theorem \ref{thm:base_walk}. It simply says the following. Suppose the current state of the base-exchange Markov chain  after $T$ steps, starting from any state $B_0 \in  \Pi(\mathcal{B}_R)$, is the base $B \in \Pi(\mathcal{B}_R)$, and suppose we output the polymatroid base $\alpha(B)$. If $T$ is large enough, such that we are in state $B$ with probability close to $w_\Pi(\alpha(B))$ for every $B \in \Pi(\mathcal{B}_R)$, then $\alpha(B)$ will be outputted with probability close to $w(\alpha(B))$, with $w_\Pi(\alpha(B))$ and $w(\alpha(B))$ as in \eqref{eq:poly_to_matroid}.

\begin{corollary}\label{cor:base_walk}
Let $\pi$ be the distribution over $\mathcal{B}_R$ with $\pi(\alpha) \propto w(\alpha)$, and let $\Pi_\pi$ be the distribution over $\Pi(\mathcal{B}_R)$ with $\Pi_\pi(B) \propto w_\Pi(\alpha(B))$.
Let $B \in \Pi(\mathcal{B}_R)$ and let $\Pi_{\sigma}^T = P^T(B,\cdot)$ be the distribution over $\Pi(\mathcal{B}_R)$ after $T$ steps of the base-exchange Markov chain $\mathcal{M} = (\Pi(\mathcal{B}_R),P)$.  Let $\sigma^T$ be the induced distribution over $\mathcal{B}_R$ given by
$
\sigma^T(\alpha) = \sum_{B : \alpha(B) = \alpha} \Pi_{\sigma}^T(B).
$
If $d_{TV}(\Pi_{\sigma}^T,\Pi_\pi) \leq \epsilon$, then also $d_{TV}(\sigma^T,\pi) \leq \epsilon$.
\end{corollary}

\begin{remark}
It is possible to define a more direct Markov chain on the set of all bases of a given discrete polymatroid, and prove that this chain is rapidly mixing (also based Theorem \ref{thm:base_walk}), but this is not necessarily needed for our results (and not of interest of this work). 
\end{remark}

\section{Extension parallel congestion games}\label{sec:ep}
An \emph{extension parallel congestion game} is a symmetric congestion game in which the common strategy set $\mathcal{P}$ of the players consists of the $o,d$-paths in a (directed) extension-parallel network $G = (V,A)$ with source $o$ and target $d$.  
For two given networks $G_i = (V_i,A_i)$ with source $o_i$ and target $d_i$ for $i = 1,2$, let $G' = (V_1 \cup V_2, A_1 \cup A_2)$ be the union of $G_1$ and $G_2$. The \emph{parallel composition} of $G_1$ and $G_2$ is the network obtained by identifying $o_1$ with $o_2$, and $d_1$ with $d_2$. These nodes are the source and target of $G'$, respectively. The \emph{series composition} of $G_1$ and $G_2$ is obtained by identifying $d_1$ with $o_2$. The node $o_1$ becomes the source of $G'$, and $d_2$ its target.
An \emph{extension parallel network} either consists of i) a single arc $(o,d)$, ii) two extension-parallel networks in parallel, iii) a single arc in series with an extension parallel network. An example is given in Figure \ref{fig:ep}.
For a given extension parallel graph $G$, we use $\mathcal{P} = \{p_1,\dots,p_q\}$ to denote all $o,d$-paths in $G$. Note that for an extension parallel network  we have $q \leq |A| = m$.
\begin{figure}[h!]\begin{center}
	\begin{tikzpicture}[scale=1]
		\node (a) at (0,0.5)[circle,scale=0.7,fill=black] {};
				\node at (0,0.1)[] {$o$};
	    \node (b) at (3,0)[circle,scale=0.7,fill=black] {};
	    \node (c) at (6,0.5)[circle,scale=0.7,fill=black] {};
		\node (a2) at (0,0.5)[circle,scale=0.7,fill=black] {};
	    \node (b2) at (3,1)[circle,scale=0.7,fill=black] {};
	    \node (c2) at (6,0.5)[circle,scale=0.7,fill=black] {};
	   \node at (6,0.1)[] {$d$};
	    
	\path[every node/.style={sloped,anchor=south,auto=false}]
	(a) edge[->,very thick,bend left=10] node {} (b)
	(a) edge[->,very thick,bend left=-10] node {} (b)
	(b) edge[->,very thick,bend left=-10] node {} (c)
	(a2) edge[->,very thick,bend left=10] node {} (b2)
	(b2) edge[->,very thick,bend left=10] node {} (c2)
	(b2) edge[->,very thick,bend left=-10] node {} (c2);

	\end{tikzpicture}		\caption{Example of an extension parallel network}
		\label{fig:ep}
\end{center}
\end{figure}
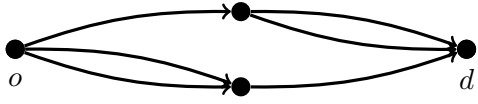 

In this section we will always be working with strategy load profiles (and so these are sometimes simply referred to as load profiles). The set of all possible strategy load profiles will be denoted by $\mathcal{L} = \{ \alpha \in \{1,\dots,q\}^n : |\alpha| = n\}$. We consider the potential $\Phi : \mathcal{L} \rightarrow \Q$ defined by $\Phi(\alpha) = \Phi(s)$ for some $s \in \mathcal{S}(\alpha)$.\footnote{This is well-defined as the potential value is the same for any choice of $s \in \mathcal{S}(\alpha)$.} 
The main result that we will need in this section is the M-convexity of Rosenthal's potential for EP congestion games.
\begin{proposition}[Fujishige et al. \cite{Fujishige2015}]
\label{prop:fujishige}
Let $\Gamma$ be an extension parallel congestion game. Then the  potential $\Phi : \mathcal{L} \rightarrow \Q$ defined by $\Phi(\alpha) = \Phi(s)$ for $s \in \mathcal{S}(\alpha)$, is $M$-convex.
\end{proposition}

Before giving our main result as sketched in Section \ref{sec:contributions}, we first give a more direct approach for sampling from the Gibbs distribution in EP congestion games.

\subsection{Sampling from the Gibbs distribution}\label{sec:ep_gibbs}
As mentioned in Section \ref{sec:overview}, and sketched in Algorithm \ref{alg:ep}, the high-level algorithmic idea for sampling a strategy profile according to the Gibbs distribution consists of first sampling a load profile $\alpha$ with the correct probability, and then a strategy profile from $\mathcal{S}(\alpha)$ uniformly at random. The main result of this section, based on this approach, is stated in Theorem \ref{thm:ep}.

\begin{theorem}
\label{thm:ep}
Let $\epsilon > 0$ and $T \geq 0$, and let $\Gamma$ be an extension parallel congestion game with $n$ players. There is a randomized algorithm $\mathcal{A}$ with output distribution $\bar{\pi}$ over $\mathcal{P}^n$ that is $\epsilon$-close to the Gibbs distribution $\pi$ at temperature $T$, and runs in (expected) time $O(C\cdot n\log(n/\epsilon))$ with $C$ the complexity of implementing one step of a base-exchange Markov chain with suitable probabilities (see Section \ref{sec:sampling}).
\end{theorem}

\begin{proof}
Note that for an extension parallel congestion game, the number of strategy profiles corresponding to a given load profile $\alpha$ is
$
|\mathcal{S}(\alpha)| = \frac{n!}{\alpha!}$.\footnote{This is the number of ways in which we can assign $n$ labeled balls to bins $b_1,\dots,b_p$, where $b_i$ contains $\alpha_i$ balls.}

\begin{lemma}\label{lem:ep_SLC}
The $n$-homogeneous generating polynomial
\begin{equation}
g(x_1,\dots,x_n) = \sum_{\alpha \in [q]^n : |\alpha| = n} |\mathcal{S}(\alpha)| e^{-T \Phi(\alpha)} x^{\alpha} = \sum_{\alpha \in [q]^n : |\alpha| = n} \frac{n!}{\alpha!} e^{-T \Phi(\alpha)} x^{\alpha}
\label{eq:ep_slc}
\end{equation}
is strongly log-concave. Hence, the distribution $\pi'$ over $\mathcal{L}$ given by $\pi'(\alpha) \propto \frac{n!}{\alpha!} e^{-T \Phi(\alpha)}$ for $\alpha \in \mathcal{L}$ is strongly log-concave.
\end{lemma}
\begin{proof}
Strong log-concavity is preserved under scalar multiplication by Proposition \ref{prop:scalar}, so it suffices to show that 
$$
\frac{1}{n!}g(x) = \sum_{\alpha \in [q]^n : |\alpha| = n} \frac{e^{-T \Phi(\alpha)}}{\alpha!}  x^{\alpha}
$$
is strongly log-concave. In turn, by Proposition \ref{prop:sufficient}, it is sufficient to show that $\log\left(e^{-T \Phi(\alpha)}\right) = -T\Phi(\alpha)$ is an $M$-concave function on its effective domain. As $T \geq 0$, this is equivalent to showing that $\Phi(\alpha)$ is $M$-convex on its effective domain $\mathcal{L} = \{ \alpha \in [q]^n : |\alpha| = n\}$. This follows from Proposition \ref{prop:fujishige}.
\end{proof}

Because of Lemma \ref{lem:ep_SLC}, the polarization $\Pi(g)$ of $g$ in \eqref{eq:ep_slc} is also strongly log-concave. The support of $\Pi(g)$ can be seen as the bases of the $n$-uniform matroid $\mathcal{N}$ on ground set $\{(i,j) : 1 \leq i,j \leq n\}$. 
Our algorithm now consists of first running the base-exchange Markov chain for 
$
O(n\log(n/\epsilon))
$
steps, starting from any initial base. We output $\alpha(B)$, where $B$ is the state we are in after the $O(n\log(n/\epsilon))$ steps that were carried out. The resulting distribution $\sigma'(\alpha)$ over $\mathcal{L}$ satisfies $d_{TV}(\sigma',\pi') \leq \epsilon$ by Corollary \ref{cor:base_walk}. We then uniformly at random choose a strategy profile from $\mathcal{S}(\alpha)$. Let $\bar{\pi}$ be the resulting output distribution over $\mathcal{S}$. 
It remains to show that $d_{TV}(\pi,\bar{\pi}) \leq \epsilon$, which can be done using a simple calculation that is deferred to Appendix \ref{app:calc}.

We conclude with analysing the running time of the algorithm. One step of the base-exchange Markov chain can be implemented in time $O(C)$ by definition. Generating an $s \in \mathcal{S}(\alpha)$ uniformly at random can be done by generating a uniform random permutation $\mu$ of $\{1,\dots,n\}$. We set $s_i = p_1$ for players $i = \mu(1),\dots,\mu(\alpha_1)$, $s_i = p_2$ for players $\mu(\alpha_1 + 1),\dots,\mu(\alpha_2 + 1)$, and so on.
Generating a uniform random permutation can be done in time $O(n\log(n))$ using random $0/1$ bits. 
\end{proof}

\subsection{Relaxed logit dynamics}\label{sec:relaxed}
In this section, we give the proof of Theorem \ref{thm:ep_informal}, reformulated in Theorem \ref{thm:relaxed} below. 
Formally, the \emph{relaxed logit dynamics Markov chain} with current state $s \in \mathcal{P}^n$ proceeds by:
\begin{itemize}
\item With probability $\frac{1}{2}$: Select two players $i,j \in N$ uniformly at random and transition to $s'$ given by 
$$
s_k' = \left\{\begin{array}{ll}
s_i & \text{ if } k = j\\
s_j & \text{ if } k = i\\
s_k & \text{ otherwise.}
\end{array}\right.
$$
\item With probability $\frac{1}{2}$: Perform a transition according to the logit dynamics (as in Section \ref{sec:pre_gibbs}).
\end{itemize}
We note that for any symmetric congestion game, this is a well-defined ergodic, time-reversible Markov chain with the Gibbs distribution as stationary distribution.

\begin{theorem}\label{thm:relaxed}
For an extension parallel congestion game $\Gamma$ with common strategy set $\mathcal{P}$ and initial state $s \in \mathcal{P}^n$, the mixing time of the relaxed logit dynamics Markov chain at temperature $T \geq 0$ satisfies
\vspace*{-0.2cm}
\begin{center}
$
\displaystyle \tau_s(\epsilon) \leq n^3 \left( \log n + \log \log |\mathcal{P}| + \log \left(\frac{2T\Phi_{\max}}{\epsilon^2}\right) \right)
$
\end{center}
where $\Phi_{max} = \max_{r \in \mathcal{S}} \Phi(r)$ is the maximum value attained by Rosenthal's potential over $\mathcal{P}^n$.
\end{theorem}
Compared to the mixing time of the (non-relaxed) logit dynamics for general games, we get a doubly exponential improvement in terms of the dependence on $T\Phi_{max}$ (at the cost of a small polynomial increase in the dependence on $n$). 

\begin{proof}[Proof of Theorem \ref{thm:relaxed}]
We will use a Markov chain decomposition argument based on the two operations that define the relaxed logit dynamics Markov chain. 
We first partition the state space $\mathcal{S} = \mathcal{P}^n$ naturally based on load profiles by setting $\Omega_{\alpha} = \mathcal{S}(\alpha)$ for $\alpha \in \mathcal{L}$, where as before we have $\mathcal{L} = \{ \alpha : \alpha \in [q]^n \text{ and } |\alpha| = n\}$. 
Our proof approach is to apply the Markov chain decomposition theorem of Hermon and Salez \cite{Hermon2019} as given in Theorem \ref{thm:decomp_all}. In particular, for this we will need to bound the modified log-Sobolev constants of the projection and restriction chains. 
We start with the modified log-Sobolev constant $\bar{\rho}$ of the projection chain.

The projection chain $\bar{P}$ has state space $\mathcal{L}$ and stationary distribution $\bar{\pi}(\alpha) = |\mathcal{S}(\alpha)|e^{-T\Phi(\alpha)}$ for $\alpha \in \mathcal{L}$. 
Let $\alpha, \beta \in \mathcal{L}$ such that $\sum_e |\alpha_e - \beta_e| = 2$, i.e., there exist  paths $p$ and $p'$ such that
$$
\alpha_e = \left\{ \begin{array}{ll}
\beta_e + 1 & \text{ if } e = p\\
\beta_e - 1 & \text{ if } e = p'\\
\alpha_e & \text{ if } e \in E \setminus \{p,p'\}.
\end{array}\right.
$$
In this case we say that $\alpha$ and $\beta$ are \emph{adjacent load profiles differing on paths $p$ and $p'$}. Note that if $s \in \mathcal{S}(\alpha)$ and $s' \in \mathcal{S}(\beta)$ are such that they differ by a deviation of some player $i$ from path $p$ to path $p'$, then 
$
P(s,s') = \frac{1}{2n}\frac{\exp(-T \Phi(p',s_{-i}))}{Z'}
$
and this expression is the same for every such player $i$ with $s_i = p$. Moreover, note that $\pi(x) = \pi(y)$ for any two strategy profiles $x,y \in \mathcal{S}(\alpha)$.

For some fixed choice of strategy profile $x \in \Omega_{\alpha}$ and player $i$ using path $p$, i.e., $s_i = p$, the transition probabilities for adjacent load profiles can then be seen to equal\footnote{Here $Z$ and $Z'$ are the normalizing constants as in Section \ref{sec:pre_gibbs}}
\begin{align}
2\bar{P}(\alpha,\beta) = \frac{1}{\bar{\pi}(\alpha)} \sum_{x \in \Omega_{\alpha}, y \in \Omega_{\beta}} \pi(x)P(x,y) &= \frac{\pi(x)}{\bar{\pi}(\alpha)}\frac{\alpha_p}{n} |\mathcal{S}(\alpha)|\frac{\exp(-T \Phi(\beta))}{Z'} \nonumber \\
&= \frac{\exp(-T\Phi(\alpha))/Z}{|\mathcal{S}(\alpha)|\exp(-T\Phi(\alpha))/Z}\frac{\alpha_p}{n} |\mathcal{S}(\alpha)|\frac{\exp(-T \Phi(\beta))}{Z'}\nonumber  \\
&= \frac{\alpha_p}{n}\frac{\exp(-T \Phi(\beta))}{Z'}\nonumber  \\
&= 2\alpha_p P(s,s') \label{eq:alpha_beta}
\end{align}
where the last equality is true for any choice of $s \in \mathcal{S}(\alpha)$ and $s' \in \mathcal{S}(\beta)$. Note that this implies that for any $\alpha, \beta \in \mathcal{L}$, $s \in \mathcal{S}(\alpha)$ and $s' \in \mathcal{S}(\beta)$, we have
\begin{equation}
\label{eq:ep_chi_bound}
\frac{P(s,s')}{\bar{P}(\alpha,\beta)} = \frac{1}{\alpha_p} \geq \frac{1}{n}.
\end{equation}
The lower bound of $1/n$ will serve as our lower bound on $\chi$ as defined in Appendix \ref{app:markov_decomp}.
In order to bound the modified log-Sobolev constant of the projection chain, one can use a comparison argument (as defined in  Appendix \ref{app:markov_comparison}) with the base-exchange Markov chain on the support of the polarization $\Pi(g_{\Gamma})$ of $g_{\Gamma}$ as in \eqref{eq:ep_slc}. In this section the support corresponds to the set of bases of an $n$-uniform matroid. In particular, it holds that 
\begin{equation}
\label{eq:ep_modsob1}
\bar{\rho} \geq \frac{1}{n}\cdot \rho(\Pi(\mathcal{L})) \geq \frac{1}{n^2},
\end{equation}
where $\rho(\Pi(\mathcal{L}))$ is the modified log-Sobolev constant of the base-exchange Markov chain on the support of $\Pi(g)$ with $g$ as in \eqref{eq:ep_slc}. The second inequality comes from the fact that $\rho(\Pi(\mathcal{L}))  \geq 1/n$, as was shown by Cryan, Guo and Mousa \cite{Cryan2019}. The first inequality is not hard to see, but we defer it to Appendix \ref{app:ep_missing} as it requires a Markov chain comparison argument between two Markov chains on different state spaces.

We continue with bounding the modified log-Sobolev constant of the restriction chains. In order to do this, we will use a comparison argument with the \emph{random transposition Markov chain} on the set $S_k$ of all permutations of $\{1,\dots,k\}$. Given a permutation $\sigma$, this chain proceeds by selecting two positions $a$ and $b$ uniformly at random, and interchanging the positions of the elements $\sigma(a)$ and $\sigma(b)$.
With $\rho_{rt}$ denoting the modified log-Sobolev constant of this chain, it follows that for every $\alpha \in \mathcal{L}$, we have
\begin{equation}
\label{eq:ep_modsob2}
\rho_{\alpha} \geq \rho_{rt} \geq \frac{1}{n-1}
\end{equation}
using the fact that $\rho_{rt} \geq 1/(n-1)$, as shown by Goel \cite{Goel2004}. This comparison argument is also deferred to Appendix \ref{app:ep_missing}. Now, applying Theorem \ref{thm:markov_comparison}, it follows that $\bar{\rho} \geq 1/n^3$.
Plugging this into \eqref{eq:pre_modsob}, it then follows that 
$$
\tau_s(\epsilon) \leq n^3 \left( \log n + \log \log |\mathcal{P}| + \log \left(\frac{2T\Phi_{\max}}{\epsilon^2}\right) \right)
$$
using that $\pi(s)^{-1} \leq |\mathcal{P}|^n e^{-T\Phi_{\max}}$ for every $s \in \mathcal{S}$, because of the non-negativity of the cost functions.
\end{proof}

\subsection{Uniform sampling of pure Nash equilibria}\label{sec:nash}
In Theorem \ref{thm:ep_nash} we show that the result in Theorem \ref{thm:ep} also implies that, for an extension parallel congestion game $\Gamma$, we can (approximate) uniformly at random sample a pure Nash equilibrium from the set $\text{NE}(\Gamma)$ of all pure Nash equilibria of $\Gamma$ in pseudo-polynomial time. That is, we sample every $s \in \text{NE}(\Gamma)$ with probability approximately $1/|\text{NE}(\Gamma)|$. The (approximate) uniform sampling of combinatorial objects has received a lot of attention in the last thirty years, in particular within the area of theoretical computer science. However, to the best of our knowledge, no non-trivial results for (pure) Nash equilibria are known, despite the fact that the problem of computing Nash equilibria has received much attention.

For the proof of Theorem \ref{thm:ep_nash} we use the fact that Nash equilibria are precisely the profiles minimizing Rosenthal's potential in EP congestion games. Furthermore, we will exploit the fact that Theorem \ref{thm:ep} holds for any temperature $T \geq 0$. In particular, if we set $T$ large enough, then most weight in the stationary distribution will be assigned to profiles minimizing Rosenthal's potential (under the assumption that the cost functions are integer-valued). This means that, with high probability, Algorithm \ref{alg:ep} will output a strategy profile minimizing Rosenthal's potential with domain $\mathcal{S}$.\footnote{Whenever we refer to Algorithm \ref{alg:ep}, we mean the implementation of the high-level approach as given in the previous section.} We will use the Gibbs distribution with base $2$ instead of $e$ to avoid having to work with real numbers. The full proof of Theorem \ref{thm:ep_nash} is given in Appendix \ref{app:ep_nash}

\begin{theorem}
\label{thm:ep_nash}
Let $\epsilon > 0$ and let $\Gamma$ be an extension parallel congestion game with integer-valued cost functions and $n$ players. There is a randomized  algorithm $\mathcal{A}$ with output distribution $\bar{\pi}$ over $\text{NE}(\Gamma)$ that is $\epsilon$-close to the uniform distribution over $\text{NE}(\Gamma)$, and runs in (expected) time polynomial in $n, m, \Phi_{\max}$ and $\log(1/\epsilon)$ where $\Phi_{\max}$ is the maximum value attained by Rosenthal's potential.
\end{theorem}

\begin{remark}
The pseudo-polynomial dependence, coming from the polynomial dependence on $\Phi_{\max}$ rather than $\log_2(\Phi_{\max})$, arises from the fact that we have to compute transition probabilities of the form $2^{a_i}/\sum_i 2^{a_i}$, where the $a_i$ are integers, which requires $\Omega(\sum_i a_i)$ random $0/1$ bits (following the notion of suitable probabilities in Section \ref{sec:sampling}). However, there is no \emph{a priori} reason that the problem of (approximately) sampling pure Nash equilibria according to the uniform distribution requires pseudo-polynomial (in the input size of the original congestion game) time, as opposed to sampling from the Gibbs distribution. We leave open the question of finding a truly polynomial time algorithm.
\end{remark}

\section{Capacitated uniform congestion games}\label{sec:cap}
In this section, we consider \emph{$u$-capacitated $k$-uniform congestion games} for vectors $u = (u_1,\dots,u_m)$ and $k = (k_1,\dots,k_n)$. We write $K = |k| = \sum_{i} k_i$ and $U = |u| = \sum_e u_e$. The vector $u$ models the capacities of the resources $e \in E$, i.e., the variables $(u_e)_{e\in E}$ as defined in Section \ref{sec:pre_cong}. The strategy set of player $i \in N$ is given by all subsets $S \subseteq E$ of cardinality $|S| = k_i$, i.e., the bases of the $k_i$-uniform matroid on $E$. We write $\Gamma(u,k)$ for the collection of all $u$-capacitated $k$-uniform congestion games. We remark that in this section \emph{load profiles} will refer to resource load profiles as defined in Section \ref{sec:pre_cong}, and no longer to \emph{path load} profiles as considered in Section \ref{sec:ep}. 

Note that we can naturally model a feasible strategy profile in $s = (s_1,\dots,s_n) \in \mathcal{S}$ of a capacitated uniform congestion game as a (simple) bipartite graph $G = (N \cup E, F) \in \mathcal{G}(u,k)$ on $N \cup E$: There is an edge $\{i,e\} \in F$ if and only if player $i \in N$ uses resource $e \in E$ in $s_i$. 

The main result needed in this section is stated in Proposition \ref{prop:mckay} below.
We use the notation $[x]_b = x(x-1)\cdots (x-b+1)$ for $x,b \in \Z_{\geq 0}$. For a bipartite degree sequence $(k,\alpha)$, we then write $K_b = \sum_{i=1}^n [k_i]_b$ and $A_b = \sum_{j=1}^m [\alpha_j]_b$. Note that $K = A = K_1 = A_1$.

\begin{proposition}[McKay \cite{McKay1984}]
Let $\mathcal{D}$ be the collection of all bipartite degree sequences $(k,\alpha)$ for which $1 \leq k_{\max}\alpha_{\max} = o\left(K^{1/4}\right)$. Then 
$$
|\mathcal{G}(k,\alpha)| = \frac{K!}{\prod_i k_i!\prod_j \alpha_j!} \exp\left(-\frac{K_2}{K^2}\cdot A_2 + O\left(\max\{k_{\max},\alpha_{\max}\}^4/K\right)\right)
$$
as $K \rightarrow \infty$.
\label{prop:mckay}
\end{proposition}

\subsection{Sampling from the Gibbs distribution}
\label{sec:cap_gibbs}
In this section we give an (almost polynomial time) sampling algorithm that samples from a distribution which is close to the Gibbs distribution (provided the game is sufficiently large). That is, we show that for a large class of pairs $(u,k)$, we can sample from a distribution close to the Gibbs distribution.

We follow again the high-level approach  in Algorithm \ref{alg:ep}. The set of all feasible load profiles is now given by 
$
L(k,u) = \left\{ \alpha : 0 \leq \alpha \leq u \text{ and } |\alpha| = \sum_i k_i \right\}
$
and $\mathcal{S}(\alpha) = \mathcal{G}(k,\alpha)$ for any feasible load profile $\alpha \in L(k,u)$.
Recall that we want to sample an $\alpha \in L(k,u)$ with probability proportional to (approximately)
$
\approx |\mathcal{S}(\alpha)| e^{-T \Phi(\alpha)},
$
and then sample a strategy profile $s \in \mathcal{S}(\alpha)$ with probability $\approx 1/|\mathcal{S}(\alpha)|$. 

A couple of problems arise here compared to the case of extension parallel congestion games. First of all, there is no polynomial time algorithm known to compute the numbers $w_{\alpha} = |\mathcal{S}(\alpha)|$.\footnote{In fact, it is still an open question  whether or not this problem is $\#$P-complete (see, e.g., \cite{JSV2004}).} Instead, we will rely on a \emph{fully-polynomial randomized approximation scheme} for computing approximations $\hat{w}_{\alpha}$ to the numbers $w_{\alpha}$ up to arbitrary precision \cite{Bezakova2007,JSV2004}.
Secondly, in this case the polynomial 
\begin{equation}
g(x_1,\dots,x_n) = \sum_{\alpha \in L(k,u)} |\mathcal{S}(\alpha)| e^{-T \Phi(\alpha)} x^{\alpha} 
\label{eq:uc_not_slc}
\end{equation}
is in general not strongly log-concave. 
We overcome this problem by showing that we can restore strong log-concavity `approximately' when the game becomes large, and when there are, in addition, suitable capacity constraints. We do this by using asymptotic enumeration formulas for the number of bipartite graphs with a given degree sequence, an area that has received considerable attention in combinatorics. It turns out that replacing $|\mathcal{S}(\alpha)|$ by an asymptotic approximation $\phi(\alpha)$ in \eqref{eq:uc_not_slc} gives rise to a strongly log-concave polynomial.

Finally, the problem of sampling a strategy profile $s \in \mathcal{S}(\alpha) = \mathcal{G}(k,\alpha)$ now corresponds to that of sampling a bipartite graph with degree sequence $(k,\alpha)$ for which many algorithms are known. 
The main result of this section is given in Theorem \ref{thm:uc_main}. 

\begin{theorem}\label{thm:uc_main}
Let $\epsilon \geq 0$, let $\pi$ be the Gibbs distribution at temperature $T \geq 0$, and let $\mathcal{D}$ be the the class of all congestion games $\Gamma(k,u)$ satisfying 
\begin{equation}
\label{eq:mckay}
1 \leq k_{\max}u_{\max} = o\left(K^{1/4}\right).
\end{equation}
There is a randomized algorithm $\mathcal{A}$ for the class $\mathcal{D}$, and a constant $K_0 \geq 0$, such that the output distribution $\bar{\sigma}$ over $\mathcal{S}$ has the property that
$$
d_{TV}(\bar{\sigma},\pi) \leq \epsilon
$$
whenever $K \geq K_0$. The algorithm runs in (expected) time 
$$
C \cdot n \left( \log n + \log \log |\mathcal{P}| + \log \left(\frac{2T\Phi_{\max}}{\epsilon^2}\right) \right)
$$
with $C(n,m,\epsilon,\Phi_{\max}) = \text{\emph{poly}}(1/\epsilon,n,m,\Phi_{\max})$.
\end{theorem}
\begin{proof}
Setting 
$$
\phi(\alpha) = \frac{K!}{k!\alpha!} \exp\left(-\frac{K_2}{K^2}\cdot A_2\right)
$$
it follows, assuming \eqref{eq:mckay} holds, that for any $0 \leq \alpha \leq u$, we have $\phi(\alpha) = (1 + o(1))|\mathcal{S}(\alpha)|$, where $o(1)$ is with respect to $K \rightarrow \infty$. In particular if $K \geq K_0$ for $K_0$ large enough, it follows that 
\begin{equation}
\label{eq:asymptotic}
\frac{1}{2}|\mathcal{S}(\alpha)| \leq \phi(\alpha) \leq \frac{3}{2}|\mathcal{S}(\alpha)|.
\end{equation}
The next step is now to show that replacing $|\mathcal{S}(\alpha)|$ by $\phi(\alpha)$ in \eqref{eq:uc_slc} gives rise to a strongly log-concave polynomial. 
The crucial observation here is to see that $A_2 = \sum_j \alpha_j (\alpha_j - 1)$ is a separable convex function.

\begin{lemma}\label{lem:uc_slc}
The $K$-homogeneous generating polynomial
\begin{equation}
g(x_1,\dots,x_n) = \sum_{\alpha \in L(k,u)} \frac{K!}{k!\alpha!} \cdot  \exp\left(-\frac{K_2}{K^2}\cdot \sum_{j=1}^m \alpha_j(\alpha_j-1) \right) \exp\left(-T \Phi(\alpha)\right) \cdot x^{\alpha}
\label{eq:uc_slc}
\end{equation}
is strongly log-concave.
\end{lemma}
\begin{proof}
Following the proof of Lemma \ref{lem:ep_SLC}, first observe that\footnote{Note that $k$ and all quantities involving $k$ are considered fixed.}
\begin{small}$$
\frac{k!}{K!} g(x_1,\dots,x_n) = \sum_{\alpha \in L(k,u)} \frac{1}{\alpha!} \exp\left(-\frac{K_2}{K^2}\cdot \sum_{j=1}^m \alpha_j(\alpha_j-1) \right) \exp\left(-T \Phi(\alpha)\right) x^{\alpha}
$$\end{small}
is strongly log-concave as well because of Proposition \ref{prop:scalar}. Then, in order to apply Proposition \ref{prop:sufficient}, it suffices to show that
$$
-\left( \frac{K_2}{K^2}\cdot \sum_{j=1}^m  \alpha_j (\alpha_j - 1) + T \Phi(\alpha) \right)
$$
is $M$-concave over its domain $L(k,u)$.\footnote{That is, formally speaking, we define it to be $-\infty$ outside of $L(k,u)$} This follows directly from the fact that both $\sum_{j=1}^m \alpha_j (\alpha_j - 1)$ and $\Phi(\alpha)$ are separable convex functions in $(\alpha_1,\dots,\alpha_m)$ over the (effective) domain $L(k,u)$, and the fact that $K_2, K, T \geq 0$. Separable convex functions (over effective domain $L(k,u)$) are known to be $M$-convex \cite{Murota2003}.
\end{proof}

One can now carry out similar steps as in the proof of Theorem \ref{thm:ep}, albeit with some modifications. Again, the polarization $\Pi(g)$ of $g$ as in \eqref{eq:uc_slc} is strongly log-concave as well. The support of $\Pi(g)$ is now the set of bases of the $n$-truncation of a partition matroid $\mathcal{N}$ on ground set $E = \cup_j E_j$ where $E_j = \{(j,i) : 1 \leq i \leq n\}$ with $A \subseteq E$ is independent if and only if $|A \cap E_j| \leq u_j$ for $j = 1,\dots,m$. It follows that the modified log-Sobolev constant of this chain satisfies $\rho_{\mathcal{N}} \geq 1/n$. A simple Markov chain comparison argument (as described in Appendix \ref{app:markov_comparison}), in combination with \eqref{eq:asymptotic}, then yields that the modified log-Solev constant $\rho$ of the Markov chain in which we use the original quantities $|\mathcal{S}(\alpha)|$, instead of the approximation $\phi(\alpha)$, satisfies $\rho \geq \frac{1}{2n}$.

Of course, the algorithmic problem is now that we cannot compute the quantities $w_{\alpha} = |\mathcal{S}(\alpha)|$ exactly in polynomial time, so one step of this base-exchange Markov chain cannot be implemented efficiently. Nevertheless we can use the approximation scheme of Bez\'akov\'a, Bhatnagar and Vigoda \cite{Bezakova2007} to compute approximations $\hat{w}_{\alpha}$ up to arbitrary precision in time polynomial in $n, m$ and $1/\epsilon$ (this gives the dependence of $1/\epsilon$ in $C$). 
A Markov chain comparison argument then implies that the base-exchange Markov chain on $\mathcal{N}$ using these approximation is also rapidly mixing, in particular, it is sufficient to run the chain  
$$
3n \left( \log n + \log \log |\mathcal{P}| + \log \left(\frac{2T\Phi_{\max}}{\epsilon^2}\right) \right)
$$
steps and then output $\alpha(B)$ where $B$ is the current base after having run the chain for the above-mentioned number of steps.

We conclude with the sampling of a strategy profile from the set $\mathcal{S}(\alpha)$ uniformly at random, which now requires sampling a bipartite graph from the set $\mathcal{G}(k,\alpha)$ uniformly at random. One algorithm to do this for degree sequences satisfying the condition in Proposition \ref{prop:mckay} is that of Arman, Gao and Wormald \cite{Arman2019} that runs in expected polynomial time.\footnote{For an overview of algorithms that can be used to (approximately) sample a bipartite graph with a given degree sequence, see, e.g., \cite{Dyer2020}.}
\color{black}
\end{proof}
\begin{remark}\label{rem:epsilon}
We remark here that the algorithm of Theorem \ref{thm:uc_main} does not run in polynomial time, as described in Section \ref{sec:sampling}, because of the dependence of $C$ on $1/\epsilon$ (as opposed to the required $\log(1/\epsilon)$). This dependence arises because of the algorithmic approximations of the numbers $|\mathcal{S}(\alpha)|$ used in the proof of Theorem \ref{thm:uc_main}. Alternatively, we could just use the approximations $\phi(\alpha)$ straight away. However, these predictions only become accurate when $K \rightarrow \infty$, so this gives a weaker result in terms of closeness to the Gibbs distribution. 
\end{remark}

\paragraph{Acknowledgements.} The author is grateful to Prasad Tetali for pointing him to \cite{Hermon2019}, and to the anonymous reviewers of EC 2021 for their useful comments.

\newpage
\bibliographystyle{plain}
\bibliography{references}

\appendix
\newpage
\section{Markov chains and functional inequalities}\label{app:markov}
Let $\mathcal{M} = (\Omega,P)$ be a time-reversible Markov chain with stationary distribution $\pi$, and $f,g : \Omega \rightarrow \R_{\geq 0}$. Let 
$\mathbb{E}_{\pi}(f) = \sum_{x \in \Omega} \pi(x)f(x)$ and 
$$
\text{Var}_{\pi}(f) = \sum_{x \in \Omega} \pi(x)(f(x) - \mathbb{E}_{\pi}(f))^2.
$$
Furthermore, define the entropy-like quantity
$$
\text{Ent}_{\pi}(f) = \mathbb{E}_{\pi}\left[f\log(f) - f\log(\mathbb{E}_{\pi}(f))\right]
$$
and the \emph{Dirichlet form}
\begin{align*}
\mathcal{E}_P(f,g) &= \frac{1}{2}\sum_{x \in \Omega} \sum_{y \in \Omega} \pi(x)P(x,y)[f(x) - f(y)][g(x) - g(y)].
\end{align*}
The Poincar\'e constant is defined by 
$$
\lambda(P) = \inf\left\{ \frac{\mathcal{E}_P(f,f)}{\text{Var}_{\pi}(f)} \ \Big| \ f: \Omega \rightarrow \R_{\geq 0}, \ \text{Var}_{\pi}(f) \neq 0 \right\}, 
$$
the \emph{log-Sobolov constant} by
$$
\alpha(P) = \inf\left\{ \frac{\mathcal{E}_P(\sqrt{f},\sqrt{f})}{\text{Ent}_{\pi}(f)} \ \Big| \ f: \Omega \rightarrow \R_{\geq 0}, \ \text{Ent}_{\pi}(f) \neq 0 \right\},
$$
amd the \emph{modified log-Sobolev constant} of the Markov chain $\mathcal{M}$ is defined by
$$
\rho(P) = \inf\left\{ \frac{\mathcal{E}_P(f,\log(f))}{\text{Ent}_{\pi}(f)} \ \Big| \ f: \Omega \rightarrow \R_{\geq 0}, \ \text{Ent}_{\pi}(f) \neq 0 \right\}.
$$
These quantities satisfy (see Prop. 3.6 in \cite{Bobkov2006})
\begin{equation}
\label{eq:comp}
2\lambda(P) \geq \rho(P) \geq 4\alpha(P).
\end{equation}
It is well-known they can be used to upper bound the mixing time of a Markov chain as
\begin{align*}
\tau_x(\epsilon) \leq \frac{1}{2\lambda(P)}\left( \log \pi(x)^{-1} + 2\log\left(\frac{1}{2\epsilon} \right)\right), \ \  \tau_x(\epsilon) \leq \frac{1}{\rho(P)}\left( \log \log \pi(x)^{-1} + \log\left(\frac{1}{2\epsilon^2} \right)\right),\  \text{and}\\
 \tau_x(\epsilon) \leq  \frac{1}{4\alpha(P)}\left( \log \log \pi(x)^{-1} + \log\left(\frac{1}{2\epsilon^2} \right)\right). \ \ \ \ \ \ \ \ \ \ \ \ \ \ \ \ \ \ \ \ \ \ 
\end{align*}

\subsection{Markov chain decomposition}\label{app:markov_decomp}
Let $\mathcal{M} = (\Omega,P)$ be a time-reversible Markov chain with stationary distribution $\pi$ and let $\Omega = \Omega_1 \cup \dots \cup \Omega_m$ be a disjoint partition of the state space. We write $\lambda$ for the Poincar\'e constant of this chain, $\bar{\lambda}$ for the Poincar\'e constant of the projection chain, and $\lambda_i$ for that of the restriction chain $P_i$ (and define $\lambda_{\min} = \min_i \lambda_i$). We define $\alpha,\ \bar{\alpha}$ and $\alpha_{\min}$ similarly for the log-Sobolev constant, as well as $\rho,\ \bar{\rho}$ and $\rho_{\min}$ for the modified log-Sobolev constant.

Hermon and Salez \cite{Hermon2019} recently gave a  Markov chain decomposition theorem that applies to the Poincar\'e constant, the log-Sobolev constant and the modified log-Sobolev constant.\footnote{For other Markov chain decomposition theorems, see, e.g., the work of Jerrum et al. \cite{Jerrum2004} (who, in particular, give stronger theorems for the Poincar\'e and log-Sobolev constant).} We next describe the necessary objects to formulate their result.

Assume that for each $i,j \in [m]$ with $i \neq j$ and $\bar{P}(i,j) > 0$, we are given a coupling $\kappa_{ij}: \Omega_i \times \Omega_j \rightarrow [0,1]$ of the probability distributions $\pi_i$ and $\pi_j$. That is, $\kappa_{ij}$ is such that
\begin{align*}
\forall x \in \Omega_i,\ \sum_{y \in \Omega_j} \kappa_{ij}(x,y) = \pi_i(x),\\
\forall y \in \Omega_j,\ \sum_{x \in \Omega_i} \kappa_{ij}(x,y) = \pi_j(y).
\end{align*}
Based on the couplings $\kappa_{ij}$, we define
$$
\chi = \min_{x \in \Omega_i, y \in \Omega_j, i,j \in [m]} \left\{ \frac{\pi(x)P(x,y)}{\bar{\pi}(i)\bar{P}(x,y)\kappa_{ij}(x,y)}\right\},
$$
with the range taken over all combinations for which the denumerator in the fraction is strictly positive. We state (a small variation of) the theorem of Hermon and Salez \cite{Hermon2019} for the modified log-Sobolev constant (for the other constants the statements are similar).

\begin{theorem}[Hermon and Salez \cite{Hermon2019}]
\label{thm:decomp_all}
With the above notation, it holds that
$
\rho \geq \min\{\chi \bar{\rho}, \rho_{min}\} 
$
with 
$$
\chi =  \max_{x \in \Omega_i, y \in \Omega_j, i,j \in [m] : \bar{P(i,j)} > 0} \left\{ \frac{P(x,y)}{\bar{P}(i,j)}, \frac{P(y,x)}{\bar{P}(j,i)} \right\}.
$$
\end{theorem}

\subsection{Markov chain comparison}\label{app:markov_comparison}
Another useful property of proving mixing time bounds through Poincar\'e and (modified) log-Sobolev constants, is that it is easy to see that small perturbations in the transition probabilities and the stationary distribution only result in mild variations in these constants, by means of a Markov chain comparison argument. Goel \cite{Goel2004} states the following for the modified log-Sobolev constant, based on similar results for the other constants by Diaconis and Saloff-Coste \cite{Diaconis1996}.

\begin{theorem}[Lemma 4.1 \cite{Goel2004}]\label{thm:markov_comparison} Let $\mathcal{M} = (\Omega,P)$ and $\mathcal{M}' = (\Omega',P')$ be two finite, reversible Markov chains with stationary distributions $\pi$ and $\pi'$, respectively, and modified log-Sobolev constant $\rho$ and $\rho'$, respectively. 
Assume there is a mapping $\phi: W(\Omega,\pi) \rightarrow W'(\Omega',\pi')$ mapping $f \rightarrow f'$ for $f : \Omega \rightarrow \R_{\geq 0}$, and constants $C, c > 0$ and $B \geq 0$ such that for all $f \in W(\Omega,\pi)$, we have 
$$
\mathcal{E}_{P'}(f',\log f') \leq C \cdot \mathcal{E}_{P}(f, \log f) \ \ \text{ and } \ \ c \cdot \text{Ent}_{\pi}(f) \leq \text{Ent}_{\pi'}(f') + B\cdot \mathcal{E}_{P}(f, \log f). 
$$
Then 
$$
\frac{c \rho'}{C + B\rho'} \leq \rho.
$$
\end{theorem}
In particular, if $\Omega = \Omega'$ and there exists a $\delta > 0$ such that $(1-\delta)P(x,y) \leq P'(x,y) \leq (1+\delta)P(x,y)$ for all $x,y \in \Omega$, and $(1-\delta)\pi(x) \leq \pi'(x) \leq (1 + \delta)\pi(x)$ for $x \in \Omega$, it directly follows that 
$$
\frac{1}{\rho} \leq \frac{1+\delta}{1-\delta}\cdot \frac{1}{\rho'}.
$$

\newpage
\section{Omitted proofs}\label{app:missing}
In this appendix we give all proofs missing from the main body. These result are not very difficult to prove, but somewhat tedious to formally write down, and, hence, were omitted.

\subsection{Proof of Corollary \ref{cor:base_walk}}
\begin{rtheorem}{Corollary}{\ref{cor:base_walk}}
Let $\pi$ be the distribution over $\mathcal{B}_R$ with $\pi(\alpha) \propto w(\alpha)$, and let $\Pi_\pi$ be the distribution over $\Pi(\mathcal{B}_R)$ with $\Pi_\pi(B) \propto w_\Pi(\alpha(B))$.
Let $B \in \Pi(\mathcal{B}_R)$ and let $\Pi_{\sigma}^T = P^T(B,\cdot)$ be the distribution over $\Pi(\mathcal{B}_R)$ after $T$ steps of the base-exchange Markov chain $\mathcal{M} = (\Pi(\mathcal{B}_R),P)$.  Let $\sigma^T$ be the induced distribution over $\mathcal{B}_R$ given by
$
\sigma^T(\alpha) = \sum_{B : \alpha(B) = \alpha} \Pi_{\sigma}^T(B).
$

If $d_{TV}(\Pi_{\sigma}^T,\Pi_\pi) \leq \epsilon$, then also $d_{TV}(\sigma^T,\pi) \leq \epsilon$.
\end{rtheorem}
\begin{proof}
We have that
\begin{align*}
2d_{TV}(\sigma,\pi) & = \sum_{\alpha \in \mathcal{B}_R} \left| \sum_{B : \alpha(B) = \alpha} \sigma'(B) - w(\alpha) \right| \\
&= \sum_{\alpha \in \mathcal{B}_R} \left| \sum_{B : \alpha(B) = \alpha} \left[ \sigma'(B) - \binom{d}{\alpha(B)}^{-1} w(\alpha)\right] \right| \ \ (\text{using } \eqref{eq:scaling}) \\
&\leq  \sum_{B \in \mathcal{B}'} |\sigma'(B) - w'(\alpha(B))| \ \ \ \ \ \ \ \ \ \ \ \ \ \ \ \ \ \ \ \ \ \ \ \ \ \ \ (\text{triangle inequality})  \\
& = 2d_{TV}(\sigma',\pi') \leq 2\epsilon.
\end{align*}
This gives the desired result.
\end{proof}

\subsection{Calculation that $d_{TV}(\pi,\bar{\pi}) \leq \epsilon$ in proof of Theorem \ref{thm:ep}} \label{app:calc}
Note that 
$$
\bar{\pi}(s) = \frac{\alpha!}{n!}\sigma'(\alpha)
$$
where $\alpha = \ell(s)$ is the load profile corresponding to strategy $s$. Then
\begin{align*}
 \sum_{s \in \mathcal{S}} |\bar{\pi}(s) - \pi(s)|  &= \sum_{\alpha} \sum_{s \in \mathcal{S}: \ell(s) = \alpha}  |\bar{\pi}(s) - \pi(s)|\\
 & = \sum_{\alpha} \sum_{s \in \mathcal{S}: \ell(s) = \alpha} \left|\frac{\alpha!}{n!}\sigma'(\alpha) - e^{-T\Phi(\alpha)}\right| \\
  & = \sum_{\alpha} \frac{n!}{\alpha!} \left|\frac{\alpha!}{n!}\sigma'(\alpha) - e^{-T\Phi(\alpha)}\right| \\
    & = \sum_{\alpha}  \left|\sigma'(\alpha) - \frac{n!}{\alpha!}e^{-T\Phi(\alpha)}\right| \\
 &\leq 2 \epsilon.
\end{align*}
This shows that $d_{TV}(\pi,\bar{\pi}) \leq \epsilon$ as desired.

\subsection{Comparison arguments omitted in proof of Theorem \ref{thm:relaxed}}\label{app:ep_missing}
\paragraph{First inequality in \eqref{eq:ep_modsob1}.}
We start with showing the first inequality in \eqref{eq:ep_modsob1}, which is
$$
\bar{\rho} \geq \frac{1}{n}\cdot \rho(\Pi(L)),
$$
by using Theorem \ref{thm:markov_comparison}. We will heavily abuse notation and write $\Pi(\cdot)$ for many different objects in order not to overload the notation. 
Remember that $\bar{\rho}$ is the modified log-Sobolev constant of the projection chain $\mathcal{M} = (L,\bar{P})$, where $L = \{\alpha \in [q]^n : |\alpha| = n\}$, with stationary distribution given by 
$$
\bar{\pi}(\alpha) = |\mathcal{S}(\alpha)|e^{-T\Phi(\alpha)}.
$$
Furthermore, $\rho(\Pi(L))$ is the modified log-Sobolev constant of the base-exchange Markov chain $\mathcal{M}_{\Pi} = (\Pi(L),P_{\Pi})$ on $\Pi(L)$ with stationary distribution $\pi_{\Pi}$. Similar to what was explained in Section \ref{sec:polymatroid}, we may write $\Pi(L) = \cup F_i$ where $F_i = \{(i,j) : 1 \leq j \leq n\}$ for $i = 1,\dots,n$, and $\Pi(L)$ is then the set of bases of the $n$-uniform matroid on $\cup F_i$. Roughly speaking, for every path $p \in \mathcal{P}$ we introduce $n$ auxiliary elements (corresponding to the $n$ auxiliary variables introduced when polarizing). 

For every $\alpha \in L$, there are $\binom{n}{\alpha}$ bases $B \in \Pi(L)$ corresponding to it (as in Section \ref{sec:polymatroid}). We will denote this set of bases by
$$
\Pi(\alpha) = \{A \in \Pi(L) : \alpha(A) = \alpha\},
$$
where $\alpha(A)$ is the vector given by $\alpha_i = |A \cap F_i|$ for $i = 1,\dots,n$.

We start with defining the required mapping $\phi$ needed in Theorem \ref{thm:markov_comparison}. For $f : L \rightarrow \R_{\geq 0}$, we define
$f' : \Pi(L) \rightarrow \R_{\geq 0}$ simply by setting $f'(A) = f(\alpha(A))$ for $A \in \Pi(L)$. It  can then easily be checked that $\text{Ent}_{\pi}(f) = \text{Ent}_{\pi_\Pi}(f')$ as $\pi(\alpha) = \sum_{A \in \Pi(\alpha) } \pi_\Pi(A)$. This means that we can take $c = 1$ and $B = 0$ in Theorem \ref{thm:markov_comparison}. In order to show the desired Dirichlet form inequality in the statement of Theorem \ref{thm:markov_comparison}, it suffices to prove that for any adjacent $\alpha, \beta \in L$, it holds that
\begin{equation}
\label{app:comp1}
\sum_{A \in \Pi(\alpha)} \pi_{\Pi}(A) \sum_{B \in \Pi(\beta)} P_{\Pi}(A,B) \leq C \cdot \bar{\pi}(\alpha)\bar{P}(\alpha,\beta).
\end{equation}
with $C = n$. The fact that this is sufficient follows from the observation that summing up \eqref{app:comp1} for all ordered pairs $(\alpha,\beta)$ for $\alpha,\beta \in L$ gives the desired result (in combination with the definition of $f'$).
 
Now, fix $\alpha, \beta \in L$ and assume that they are adjacent (the case $\alpha = \beta$ can be dealt with similarly). Remember that $\gamma\in L$ is adjacent to $\alpha$ if $\sum_e |\alpha_e - \gamma_e| = 2$, i.e., there exist  paths $p$ and $p'$ such that
$$
\alpha_e = \left\{ \begin{array}{ll}
\gamma_e + 1 & \text{ if } e = p\\
\gamma_e - 1 & \text{ if } e = p'\\
\alpha_e & \text{ if } e \in E \setminus \{p,p\}.
\end{array}\right.
$$
Let $r$ be the path for which $\alpha_{r} = \beta_{r} + 1$ and write $\mathcal{N}_r(\alpha)$ for all load profiles $\gamma$ adjacent to $\alpha$ for which $\alpha_r = \gamma_r + 1$ (including $\beta$). Following the definition of the base-exchange Markov chain, it then holds that for any $A \in \Pi(\alpha)$ and $B \in \Pi(\beta)$, we have
\begin{equation}
\label{app:comp2}
2\cdot \sum_{B \in \Pi(\beta)} P_{\Pi}(A,B) = \frac{\alpha_p}{n} \frac{(n - \beta_{p'_\beta} + 1)\binom{n}{\beta}^{-1}e^{-T\Phi(\beta)}}{\sum_{\gamma \in \mathcal{N}_r(\alpha) \cup \{\alpha\}} (n - \gamma_{p'_\gamma} + 1)\binom{n}{\gamma}^{-1}e^{-T\Phi(\gamma)}} =: Q
\end{equation}
where $p'_\gamma$ is used to indicate the path $p' = p'_{\gamma}$ for which $\alpha_{p'} = \gamma_{p'} - 1$ for $\gamma \in \mathcal{N}_p(\alpha)$, and $p'_\alpha = r$. With some care, it can be shown that for $\gamma \in \mathcal{N}(\alpha) \cup \{a\}$, it holds that
$$
 \frac{(n - \gamma_{p'_\gamma} + 1)\binom{n}{\gamma}^{-1}}{(n - \beta_{p'_\beta} + 1)\binom{n}{\beta}^{-1}} = \frac{\gamma_{p'_\gamma}}{\beta_{p'_\beta}} = \frac{\beta_{p'_\gamma} + 1}{\beta_{p'_\beta}} \geq \frac{1}{n}.
$$
Continuing the estimate in \eqref{app:comp2}, we then get
$$
Q \leq n \cdot \frac{\alpha_p}{n}  \frac{e^{-T\Phi(\beta)}}{\sum_{\gamma \in \mathcal{N}_r(\alpha) \cup \{\alpha\}} e^{-T\Phi(\gamma)}} = 2n \bar{P}(\alpha,\beta),
$$
using \eqref{eq:alpha_beta} for the final equality. This gives the desired result in \eqref{app:comp1}. Applying Theorem \ref{thm:markov_comparison} with $c = 1,\ C = n$ and $B = 0$ then gives the desired first inequality in \eqref{eq:ep_modsob1}.

\paragraph{First inequality in \eqref{eq:ep_modsob2}.} In order to show the inequality in  \eqref{eq:ep_modsob2}, we will again use a Markov chain comparison between two chains on different state spaces. We want to show that
$$
\rho_{\alpha} \geq \rho_{rt} 
$$
where $\rho_{\alpha}$ is the modified log-Sobolev constant of the restriction chain on $\mathcal{S}(\alpha)$ for $a \in L$, in which we randomly interchange the strategies of two players, and $\rho_{rt}$ the modified log-Sobolev constant of the so-called random transposition walk. From now on, we fix some $\alpha \in L$.

It is convenient to study these chains in terms of bipartite graphs  with given degrees on node partition $A \cup B$.  For $\alpha \in L$ we consider the degree sequence $x = (x_1,\dots,x_n)$ with $x_i = 1$ for every $i \in B$, and the sequence $y = (y_1,\dots,y_q)$ with $y_p = \alpha_p$ for $p \in A$, where one should remember that $q$ is the number of strategies, i.e., paths, available in the common strategy set (denoted by $A$ here) of all players. It follows directly that there is a one-to-one correspondence between $\mathcal{S}(\alpha)$ and $\mathcal{G}(x,y)$ where, for a given strategy profile $s \in \mathcal{S}(\alpha)$, there is an edge $\{i,p\}$ if and only if $s_i = p$. (This is similar to the setting we consider in Section \ref{sec:cap}.) Given $s \in \mathcal{S}(\alpha)$, our restriction chain can be interpreted as randomly selecting two edges from the bipartite graph $G_s$ corresponding to the profile $s$ and switching them if possible. That is, if we select $\{i,p\}$ and $\{i',p'\}$ with $p \neq p'$, we delete the edges $\{i,p\}$ and $\{i',p'\}$, and add the edges $\{i,p'\}$ and $\{i',p\}$ (note that $i \neq i'$ always holds as the nodes in $B$ have degree one).

In order to introduce the random transposition Markov chain, we split up every node $p_j \in A$ into nodes $p_{j1},\dots,p_{j\alpha_j}$, and consider bipartite graphs on two sets of $n$ nodes $B = \{1,\dots,n\}$ and $A^* = \cup_j A^*_j$, with $A^*_j = \{p_{j1},\dots,p_{j\alpha_j}\}$, where every node has degree one. That is, every such graph is a perfect matching between $A^*$ and $B$. Note that there are precisely
$$
\prod_{j=1}^q \alpha_j! = \alpha!
$$
perfect matchings corresponding to the graph $G_s$ for $s \in \mathcal{S}(\alpha)$ under the natural transformation in which, for a given perfect matching, we consider the graph that we get by merging all the nodes $p_{j1},\dots,p_{j\alpha_j}$ back into one node $p_j$ for every $j = 1,\dots,q$. We will denote this set of perfect matchings by $H(s)$ for $s \in \mathcal{S}(\alpha)$. The random tranposition Markov chain $\mathcal{M} = (\mathcal{H},P)$, with $\mathcal{H} = \cup_s H(s)$ denoting the set of all perfect matchings on the bipartition $A^*\cup B$, proceeds by selecting two edges (of the current perfect matching) uniformly at random, and switching them. Note that this is always possible here as opposed to in the case of our restriction chains on $\mathcal{S}(\alpha)$.

We can now use a similar type of comparison argument as for the first inequality in \eqref{eq:ep_modsob1} given above. We define the mapping $\phi$, for a given function $f : \mathcal{S}(\alpha) \rightarrow \R_{\geq 0}$, by setting $f'(M) = f(s)$ whenever $M \in H(s)$ for $s \in \mathcal{S}(\alpha)$. Let $\sigma$ be the uniform distribution over $\mathcal{S}(\alpha)$ and let $\sigma_{\mathcal{H}}$ the uniform distribution over  $\mathcal{H}$. Note that $\sigma(s) = \sum_{M \in H(s)} \sigma_{\mathcal{H}}(M)$.
It then follows that  for every $s, s' \in \mathcal{S}(\alpha)$, we have
\begin{equation}
\label{app:comp3}
\sum_{M \in H(s)} \sigma_{\mathcal{H}}(M) \sum_{M' \in H(s')} P(M,M') = \sigma(s)P_{\alpha}(s,s')
\end{equation}
since $P_{\alpha}(s,s') = \frac{1}{n(n-1)} = \sum_{M' \in H(s')} P(M,M')$ for all $M \in H(s)$ whenever $s \neq s'$ and $P_{\alpha}(s,s') > 0$. Note that there is only one matching $M' \in H(s')$ such that $P(M,M') > 0$. When $s = s'$, we also have
$$
\sum_{M' \in H(s')} P(M,M') = P_{\alpha}(s,s').
$$
This implies that we can take $C = 1$. As before, we can take $a = 1$ and $B = 0$.

\subsection{Proof of Theorem \ref{thm:ep_nash}}
\label{app:ep_nash}
\begin{rtheorem}{Theorem}{\ref{thm:ep_nash}}
Let $\epsilon > 0$ and let $\Gamma$ be an extension parallel congestion game with integer-valued cost functions and $n$ players. There is a randomized  algorithm $\mathcal{A}$ with output distribution $\bar{\pi}$ over $\text{NE}(\Gamma)$ that is $\epsilon$-close to the uniform distribution over $\text{NE}(\Gamma)$, and runs in (expected) time polynomial in $n, m, \Phi_{\max}$ and $\log(1/\epsilon)$ where $\Phi_{\max}$ is the maximum value attained by Rosenthal's potential.
\end{rtheorem}\\

For the proof of Theorem \ref{thm:ep_nash} we will use the following correspondence between Nash equilibria and strategy profiles minimizing Rosenthal's potential. 

\begin{proposition}[Holzman and Law-Yone \cite{Holzman1997}; Fotakis \cite{Fotakis2010}]
The set of strategy profiles $\text{NE}(\Gamma)$ of an extension parallel congestion game $\Gamma$ coincides with the set of strategy profiles that minimize Rosenthal's potential as in \eqref{eq:rosenthal}.
\label{prop:min_nash}
\end{proposition}

\begin{proof}[Proof of Theorem \ref{thm:ep_nash}]
We first show that, for $T$ sufficiently large in the algorithm used to prove Theorem \ref{thm:ep}, most weight will be assigned to strategy profiles minimizing Rosenthal's potential. We will apply the idea in Algorithm \ref{alg:ep} used to prove Theorem \ref{thm:ep} with base $2$ instead of base $e$. Remember that $q$ is the number of $(o,d)$-paths in the extension parallel network of the game $\Gamma$.
Let $\phi = \Phi(s)$ be the common potential value of all strategy profiles $s \in \text{NE}(\Gamma)$. 
For any other strategy profile $s' \in \mathcal{S} \setminus \text{NE}(\Gamma)$, we have
$$
2^{-T \Phi(s')} \leq 2^{-T (\phi + 1)} = 2^{-T} e^{-T\phi}
$$
by assumption that all cost functions are integer-valued. As there are $q$ strategies to choose from for every player, we have $|\mathcal{S}| = q^n = 2^{n\log_2(q)}$.
This implies that the Gibbs distribution $\pi$ over $\mathcal{S}$  with temperature $T = \lceil n\log_2(q) + \log_2(2/\epsilon)\rceil$, satisfies
\begin{equation}\label{eq:ep_weight}
\pi(\mathcal{S}\setminus \text{NE}(\Gamma)) = \sum_{s \in \text{NE}(\Gamma)} 2^{-T \Phi(s')} \leq 2^{n\log(q)} 2^{-T} 2^{-T\phi} \leq \frac{\epsilon}{2} \cdot \pi(\text{NE}(\Gamma)).
\end{equation}

The algorithm for sampling an (almost) uniform sample from $\text{NE}(\Gamma)$ now works as follows. First compute a strategy profile minimizing Rosenthal's potential in order to determine $\phi$. This can be done efficiently, see, e.g., \cite{Fotakis2010}. Then run Algorithm \ref{alg:ep} with $T = \lceil n\log_2(q) + \log_2(2/\epsilon)\rceil$ and $\epsilon' = \epsilon/2$. If the resulting strategy profile has potential value $\phi$, output this strategy profile, and, otherwise, rerun Algorithm \ref{alg:ep} until it does. Note that with probability at least $(1 - \epsilon/2)$, Algorithm \ref{alg:ep} will output a strategy profile with potential value $\phi$ in one run. A simple argument then shows that the output distribution is $\epsilon$-close to the uniform distribution over $\text{NE}(\Gamma)$ as desired. 
\end{proof}
\end{document}